\documentclass[lettersize,journal]{IEEEtran}
\usepackage{flushend}
\usepackage{cite}
\usepackage{amsmath,amssymb,amsfonts}
\usepackage{mathrsfs}
\usepackage{arydshln}
\usepackage{algorithmic}
\usepackage{graphicx}
\usepackage{textcomp}
\def\BibTeX{{\rm B\kern-.05em{\sc i\kern-.025em b}\kern-.08em
    T\kern-.1667em\lower.7ex\hbox{E}\kern-.125emX}}

\usepackage{setspace}
\usepackage{cite}
\usepackage{caption}
\usepackage{subcaption}

\newtheorem{theorem}{Theorem}
\newtheorem{lemma}{Lemma}
\newtheorem{corollary}{Corollary}
\newtheorem{definition}{Definition}
\newtheorem{remark}{Remark}
\newtheorem{example}{Example}
\newenvironment{proof}{\ {\noindent\it Proof:}}{\hfill $\blacksquare$\par}

\begin{document}
\title{Controllability of Networked Sampled-data Systems}
\author{Zixuan Yang, Xiaofan Wang, and Lin Wang
\thanks{This work was supported by the National Natural Science Foundation of China
(No. 61873167, 61773255), and in part by the Strategic Priority Research Program of Chinese Academy of Sciences under Grant No. XDA27000000. \textit{(Corresponding author: Lin Wang.)}}
\thanks{Zixuan Yang and Lin Wang are with the Department of Automation, Shanghai Jiao Tong University, and Key Laboratory of System Control and Information Processing, Ministry of Education of China, Shanghai 200240, China (e-mail: jennifer\_yang@sjtu.edu.cn; wanglin@sjtu.edu.cn). }
\thanks{Xiaofan Wang is with Shanghai Key Laboratory of Power Station Automation
Technology, School of Mechatronic Engineering and Automation,
Shanghai University, Shanghai 200444, China (e-mail: xfwang@sjtu.edu.cn).}}
\maketitle
\begin{abstract}
The controllability of networked sampled-data systems with zero-order holders on the control and transmission channels is explored, where single- and multi-rate sampling patterns are considered, respectively. 
The effects of sampling on the controllability of networked systems are analyzed, with some necessary and/or sufficient controllability conditions derived.
Different from the sampling control of single systems, the pathological sampling of node systems could be eliminated by an appropriate design of network structure and inner couplings.
While for singular topology matrices, the pathological sampling of single nodes will cause the entire system to lose controllability. 
Moreover, any periodic sampling will not affect the controllability of networked systems with specific node dynamics. 
All the results indicate that whether a networked system is under pathological sampling or not is jointly determined by mutually coupled factors.
\end{abstract}


\section{Introduction}
\label{sec:introduction}
Controllability, as a prerequisite of effective control of system performance and behavior\cite{xiang2019advances}, has been extensively investigated during the past decades with various rank conditions and graphic properties developed\cite{chen2012optimal,kalman1962canonical,davison1975new,gilbert1963controllability,kobayashi1978controllability,lin1974structural}.
Recently, the upsurge of network science has promoted the expansion of research on controllability.
The maximum matching theorem was applied to ensure structural controllability by driver nodes selection \cite{liu2011controllability}, and the relationship between the number of driver nodes and network structure properties was investigated \cite{menichetti2014network, posfai2013effect}.
The notion of target control was proposed in \cite{posfai2013effect}, and the control algorithm was optimized through preferential matching \cite{zhang2017efficient}.
Multiple descriptions of the node importance were presented \cite{liu2012control,jia2013emergence} and employed to analyze the robustness of network controllability \cite{pu2012robustness}.
By the output Gramian of infinite lattice systems, \cite{klickstein2020controllability} found the control energy is exponentially related to the maximum distance between driver and target nodes.
So far, the controllability of complex networks has been applied to real-world control problems in biology, medicine, and other fields \cite{yan2017network,guo2018novel}.\par

Due to the development of computer control platforms, and problems such as limited bandwidth, signal instability, and transmission delay, the control and transmission signals of real systems are usually sampled data.
Therefore, how to maintain the controllability of a system during sampling is a problem worthy of attention.
It was pointed out in \cite{bar1975preservation} that the controllability of linear time-invariant (LTI) systems may not be preserved during the periodic sampling.
In \cite{chen2012optimal}, the sampling that destroys controllability was defined as ``pathological sampling", and a non-pathological sampling condition related to eigenvalues of the state matrix was provided.
The non-pathological sampling of switched linear systems was studied in \cite{babaali2005nonpathological}.
The influence of sampling on controllability indices was analyzed in \cite{hagiwara1988controllability}.
In \cite{kreisselmeier1999sampling}, a method of adding a step of non-equidistant sampling to periodic sampling was proposed to maintain the controllability of LTI systems and is further applied to time-varying systems\cite{guo2005systems}.
The multi-rate sampling pattern where control channels have nonidentical sampling periods was studied in \cite{pasand2018controllability} with a sufficient controllability condition given.\par

The researches above are either for single network control systems with no sampling on transmission channels between nodes or under the assumption that each node system in the network is one-dimensional.
However, node states in real-world systems are often high-dimensional and coupled with each other through multiple transmission channels.
In \cite{zhou2015controllability}, some controllability conditions for general networked systems were presented based on the transfer function matrix.
An easier-to-verify necessary and sufficient controllability criterion was derived in \cite{hao2019new} by matrix similarity transformation. 
In \cite{wang2016controllability}, it was clarified that the controllability of a networked system is jointly determined by the coupling of network structure and node dynamics.
A controllability decomposition approach for networked systems was proposed in \cite{iudice2019node} to investigate each node system when the network is not completely controllable.
The controllability of heterogeneous networked systems was investigated in \cite{xiang2019controllability}.
The controllability conditions for networked systems with more complicated structures, such as deep-coupling networks \cite{wu2020controllability} and Kronecker product networks \cite{hao2019controllability} were also derived.\par

Though there have been easier-to-verify controllability conditions for networked continuous-time LTI (CLTI) systems, the sampling effects on them has not attracted much attention.
Some results showed that the controllability of a special type of networked systems, multi-agent systems (MASs), can be decoupled into two independent parts related to single nodes and network topology, respectively \cite{ji2014protocols,ni2013consensus,xiang2013controllability}.
Some research on the sampling controllability of MASs provided necessary and sufficient conditions under synchronous and asynchronous sampling protocols\cite{lou2012controllability,zhao2020data}.
In \cite{ji2014controllability}, it was found that MASs can always maintain controllability after sampling.
And \cite{lu2020sampled} discovered that topology switching could ensure the controllability of MASs with uncontrollable subsystems.
However, to date, there has been little work devoted to the controllability of more general networked sampled-data systems that cannot be decoupled.\par

This paper studies the state controllability of networked sampled-data systems. Based on the above analysis, we found that most existing research on the controllability of networked sampled-data systems did not consider multi-rate cases and general topologies, and only
control sampling was taken into account, without sampling on transmission channels. 
The contributions of this paper are fourfold.
\begin{itemize}
    \item[(1)] In the proposed model, sampling on both control and transmission channels is considered, including single- and multi-rate sampling patterns.
    More general networked systems with directed, weighted topology and multi-dimensional node dynamics are studied in this paper.
    \item[(2)] Necessary and/or sufficient controllability conditions are developed (mainly Theorems \ref{thm-main},\ref{thm2}, Corollaries \ref{coro1},\ref{coro_add3}) for networked sampled-data systems, which have much lower computational complexity than classic criteria.
    \item[(3)] The influence of sampling on the networked system's controllability is depicted, which is coupled with network topology, inner couplings, external inputs, and node dynamics.
    Especially, systems with one-dimensional or self-loop-state node dynamics can retain controllability during any periodic sampling (Corollaries \ref{thm:1state},\ref{coro:loop}).
    \item[(4)]The influence of network structure on sampling controllability is analyzed.
    The loss of controllability caused by the pathological sampling of single node systems can be eliminated by an appropriate design of network topology and inner couplings (Remark \ref{rem6}) unless the topology matrix is singular (e.g., undirected graph, chain, star) (Corollary \ref{coro2}).
\end{itemize}
\par
The rest of this paper is organized as follows.
The model formulation and preliminaries are introduced in Section \ref{sec:pre}.
In Section \ref{sec:Main}, some controllability conditions are developed for general networked sampled-data systems.
Networked systems with special topologies and dynamics are studied in Section \ref{sec:spe-topo} and Section \ref{sec:spe-dyna}, respectively.
Section \ref{sec:multi} preliminarily inspects the controllability of networked multi-rate sampled-data systems.
Finally, Section \ref{sec:con} summarizes this paper.

\section{Preliminaries and Model Formulation}
\label{sec:pre}
In this section, some useful preliminaries, and the model of networked systems with a general structure and periodic sampling pattern are introduced.

\subsection{Notations and Definitions}
Let $\mathbb{R}$, $\mathbb{C}$, and $\mathbb{N}$ denote the fields of real, complex, and natural numbers, respectively.
Denote by $I_n$ the identity matrix of size $n\times n$, and by $e_i$ the row vector with all zero entries except that the $i$th element is $1$. 
Denote by $diag\{a_1,a_2,...,a_n\}$ the $n\times n$ matrix with diagonal elements $a_1,a_2,...,a_n$, and by $diagblock\{A_1,A_2,...,A_n\}$ the matrix with diagonal block matrices $A_1,A_2,...,A_n$.\par
Moreover, denote the set of all eigenvalues of matrix $A\in\mathbb{R}^{n\times{n}}$ by $\sigma(A)=\{\lambda_1,...,\lambda_r\},1\leq{r}\leq{n}$, and by $M(\lambda_i|A)$ the eigenspace of $A$ with respect to $\lambda_i$.
$\mathcal{D}(A)$ denotes the dimension of the eigenspace of $A$. Let $\mathcal{R}(\cdot)$ and $\mathcal{N}^\top(\cdot)$ denote the column space and left null space, respectively.
Denote $dim(\cdot)$ the dimension of space.
The complex linear span of row vectors $v_1,v_2,...,v_n$ is denoted by $span\{v_i,v_2,...,v_n\}=\{\Sigma_{i=1}^n c_iv_i|c_i\in\mathbb{C}\}$, which is the set of their all complex linear combinations.
Let $A\otimes{B}$ denote the Kronecker product of matrices $A$ and $B$, and $V_1\odot{V_2}$ the direct sum of space $V_1$ and $V_2$.
Assumed that the dimensions of matrices are compatible with algebraic operations if they are not specified.\par

\begin{definition}\cite{roman2005advanced}\label{def1}
A row vector, $v^m$ is called an $m$th-order generalized left eigenvector of matrix $A$ corresponding to $\lambda\in\sigma(A)$ if:
$$v^m(A-\lambda{I})^m=0 \quad \text{and} \quad v^m(A-\lambda{I})^{m-1}\neq{0},$$
then $v^1,v^2,...,v^{\alpha}$ form a left Jordan chain of $A$ about $\lambda$ on top of $v^1$, where the maximum value of $\alpha$ is the length of this Jordan chain.
\end{definition}

\begin{definition}\cite{hao2019new}\label{def2}
Let $E\in\mathbb{C}^{n\times{n}},H\in\mathbb{C}^{n\times{n}},\theta\in\sigma(E)$. If row vectors $\xi^1,\xi^2,...,\xi^{\gamma}$ satisfy:
$$\xi^1(\theta{I}-E)=0,\ \text{and} \ 
\xi^j(\theta{I}-E)=\xi^{j-1}H,\ j=2,...,\gamma,$$
then $\xi^1, \xi^2,...,\xi^{\gamma}$ form a generalized left Jordan chain of $E$ about $H$ corresponding to $\theta$, where $\xi^1$ is the top vector and the maximum value of $\gamma$ is the length of this generalized Jordan chain.
\end{definition}

\subsection{Control Sampling of Single Systems}
Describe a CLTI system by $\dot{x}(t)=Ax(t)+Bu(t)$, where $x\in\mathbb{R}^n$ and $u\in\mathbb{R}^p$ represent the state and input, respectively, $A\in\mathbb{R}^{n\times{n}}$, $B\in\mathbb{R}^{n\times{p}}$.
By periodic sampling on control channels, a corresponding sampled-data system can be obtained and described by $x((k+1)h)=e^{Ah}x(kh)+\int_0^h e^{A\tau}d\tau Bu(kh)$,
where $h$ denotes the sampling period, $k\in\mathbb{N}$.\par
According to the Popov-Belevitch-Hautus (PBH) rank condition\cite{xiang2019advances}, a CLTI system $(A,B)$ is controllable if and only if $\forall{s}\in\mathbb{C}$, $rank[sI_n-A,B]=n$.
Due to the difference between the reachable subspace and controllable subspace of discrete-time LTI (DLTI) systems\cite{hespanha2018linear}, the sampled-data system is controllable if but not only if $\forall{s}\in\mathbb{C}$, $rank[sI_n-e^{Ah},\int_0^h e^{A\tau}d\tau B]=n$.\par

An uncontrollable CLTI system cannot gain controllability by sampling.
On the contrary, a controllable CLTI system may become uncontrollable after sampling due to the improper selection of $h$.
The sampling that destroys the controllability is called pathological sampling.
In \cite{chen2012optimal}, a definition of (non-)pathological sampling is provided:
\begin{definition}\cite{chen2012optimal}
The control sampling with period $h$ is non-pathological about $A$ if $\forall\lambda_1,\lambda_2\in\sigma(A)$ satisfy
$$\lambda_1-\lambda_2\neq\frac{2k\pi}{h}\rm{i}, \ k=\pm{1},\pm{2},... ,$$
otherwise $h$ is pathological about $A$.
Here, ``$\rm{i}$" refers to the imaginary unit.
\end{definition}

\subsection{General Model of Networked Sampled-data Systems}
Consider a general directed and weighted network consisting of identical node systems:
\begin{equation}
\left\{
        \begin{array}{l}
        \dot{x}_i(t)=Ax_i(t)+\sum_{j=1}^N{w_{ij}Hy_j(t)}+\delta_iBu_i(t)\\
        y_i(t)=Cx_i(t)\\
        \end{array}
\right.
\label{con:linear}
\end{equation}
$i=1,2,...,N$.
$x_i\in\mathbb{R}^n$, $y_i\in\mathbb{R}^m$ and $u_i\in\mathbb{R}^p$ denote the state, output, and external control input of node $i$, respectively.
$A\in\mathbb{R}^{n\times{n}}$ is the state matrix describing the node dynamics.
$B\in\mathbb{R}^{n\times{p}}$ is the input matrix, with $\delta_i=1$ if node $i$ is under control, otherwise $\delta_i=0$.
$C\in\mathbb{R}^{m\times{n}}$ is the output matrix, and $H\in\mathbb{R}^{n\times{m}}$ describes the inner couplings among the nodes.
$w_{ij}\neq{0}$ if an edge from node $j$ to node $i$ exists, otherwise $w_{ij}=0$.
Specially, $w_{ii}\neq{0}$ if there is a self-ring of node $i$ in the network, otherwise $w_{ii}=0$.\par

\begin{figure}[tb]
\centering
\begin{subfigure}[b]{.23\textwidth}
  \centering
  \includegraphics[width=.9\linewidth]{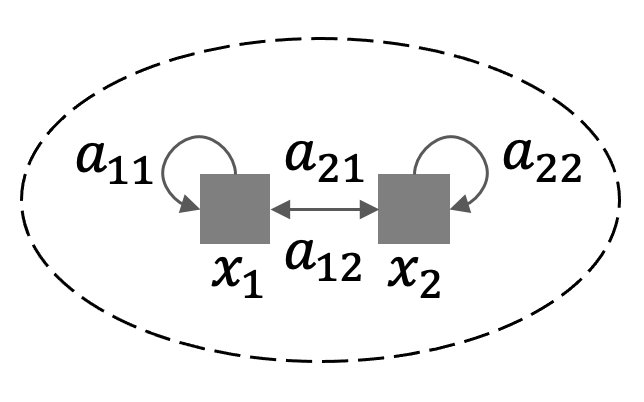}
  \caption{}
  \label{fig:node_system}
\end{subfigure}
\begin{subfigure}[b]{.23\textwidth}
  \centering
  \includegraphics[width=.83\linewidth]{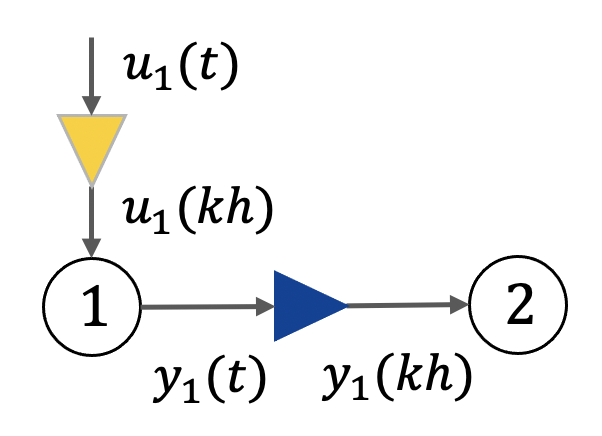} 
  \caption{}
  \label{fig:networked}
\end{subfigure}
\begin{subfigure}[b]{.47\textwidth}
   \centering
   \includegraphics[width=\linewidth]{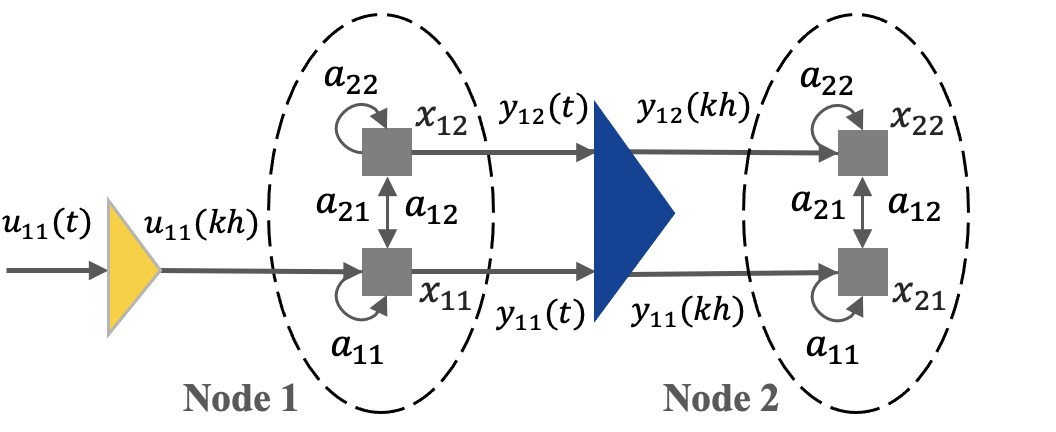}
   \caption{}
   \label{fig:all-networked}
\end{subfigure}
\caption{Networked sampled-data system consisting of two identical node systems. Compared to Fig.1 in \cite{hao2019new},
yellow and blue triangles are added to represent control and transmission sampling, respectively. (a) Node system. (b) Network structure. (c) The whole system.}
\label{fig:sampled}
\end{figure}

$W=[w_{ij}]\in\mathbb{R}^{N\times{N}}$ and $\Delta=$ diag$\{\delta_1,...,\delta_N\}\in\mathbb{R}^{N\times{N}}$ represent the transmission channels and the control channels of the networked system, respectively.
Let $X(t)=[x_1^\top(t),...,x_N^\top(t)]^\top$ be the total state of the networked system, and $U(t)=[u_1^\top(t),...,u_N^\top(t)]^\top$ be the total external control input. 
Then, the networked system can be rewritten in a compact form:
\begin{equation}
    \dot{X}(t)=\Phi{X}(t)+\Psi{U(t)}
\label{con:multi_linear}
\end{equation}
with
\begin{equation}
    \Phi=I_N\otimes{A}+W\otimes{HC}, \  \Psi=\Delta\otimes{B}.
\label{con:multi_linear_details}
\end{equation}\par
Then consider the associated sampled-data case.
Assume that the sampling is performed on all control channels and transmission channels simultaneously by samplers with zero-order hold.
Let $h$ be the sampling period, $t\in(kh,(k+1)h],k\in\mathbb{N}$.
As a result, a corresponding sampled-data system of (\ref{con:multi_linear},\ref{con:multi_linear_details}) can be represented as follows:
\begin{equation}
    X((k+1)h)=\Phi_s{X}(kh)+\Psi_s{U(kh)}
\label{MIMO_sampled_1}
\end{equation}
with
\begin{equation}
\label{MIMO_sampled_2}
    \Phi_s=I_N\otimes{e^{Ah}}+W\otimes\mathcal{H}(h), \ \Psi_s=\Delta\otimes\mathcal{B}(h),
\end{equation}
\begin{equation}
    \mathcal{B}(h)=\int_0^h e^{A\tau}d\tau{B}, \  \mathcal{H}(h)=\int_0^h e^{A\tau}d\tau{HC}.
\end{equation}

For example, a two-node networked sampled-data system is shown in Fig.\ref{fig:sampled}.
The identical two-state node system is depicted in Fig.\ref{fig:node_system}, the network structure is shown in Fig.\ref{fig:networked}, and the entire system is in Fig.\ref{fig:all-networked}. 
Let $a_{11}=a_{12}=a_{22}=1,a_{21}=-1$, and $h=\pi$, with
$$\begin{aligned}
&A=\left[\begin{array}{cc}
    1 & 1 \\
    -1 & 1
\end{array}\right],
B=\left[\begin{array}{cc}
    1 & 0 \\
    0 & 0
\end{array}\right],
C=\left[\begin{array}{cc}
    1 & 0 \\
    0 & 1
\end{array}\right],\\
&\Delta=\left[\begin{array}{cc}
    1 & 0 \\
    0 & 0
\end{array}\right],
W=\left[\begin{array}{cc}
    0 & 0 \\
    1 & 0
\end{array}\right],
H=\left[\begin{array}{cc}
    1 & 0 \\
    0 & 1
\end{array}\right].
\end{aligned}$$
By the PBH rank condition, one has $rank([sI_4-\Phi_s,\Psi_s])=2$ when $s=23.14$.
However, for the corresponding CLTI system $(\Phi,\Psi)$, $rank([sI_4-\Phi,\Psi])=4$ no matter what value $s$ takes.\par
It can be seen that the controllability of the networked system has been lost during the sampling.
However, the above verification may result in a curse of dimensionality when the network scale is large, so lower-dimensional criteria are needed.
Compared to the single LTI system, there are more coupled factors playing a role in the controllability of networked systems, such as the topologies and inner couplings, because the sampling performs on not only control channels but transmission channels.
Therefore, the controllability analysis for control sampling of single node systems cannot be directly utilized on networked sampled-data systems, either.

\section{Main Results}\label{sec:Main}
In this section, necessary and/or sufficient controllability conditions are derived for the networked sampled-data system (\ref{MIMO_sampled_1},\ref{MIMO_sampled_2}).
At first, a general condition is given in Theorem \ref{thm-main}, followed by other conditions for systems with special structures. Examples are provided to illustrate the verification.

First of all, Lemmas \ref{lem1}-\ref{chara} are given as preliminaries for the controllability rank conditions later on.

\begin{lemma}\label{lem1}
Suppose that $\sigma(W)=\{\lambda_1,\lambda_2,...,\lambda_r\},$ and $\sigma(E_i)=\{\theta_i^1,\theta_i^2,...,\theta_i^{p_i}\}$, where
\begin{equation}
    E_i=e^{Ah}+\lambda_i\mathcal{H}(h),  \nonumber
\end{equation}
$i=1,...,r.$ Then, $\sigma(\Phi_s)=\{\theta_1^1,\theta_1^2,...,\theta_1^{p_1},...,\theta_r^1,\theta_r^2,...,\theta_r^{p_r}\}$.
\end{lemma}

Lemma \ref{lem1} gives a direct representation of all the eigenvalues of $\Phi_s$, where the state matrix of the networked sampled-data system is essentially decomposed into a series of lower-dimensional matrices $E_1,...,E_r$.
Actually, $E_i$ can be regarded as the state matrix of a subsystem, which is related to an eigenvalue $\lambda_i$ of the network topology matrix $W$.\par
Moreover, let the left Jordan chain of $W$ corresponding to $\lambda_i$ be $v_i^1,v_i^2,...,v_i^{\alpha_i}$.
And let the generalized left Jordan chain of $E_i$ about $\mathcal{H}(h)$ related to $\theta_i^j$ be $\xi_{ij}^1,\xi_{ij}^2,...,\xi_{ij}^{\gamma_{ij}}$, where $\theta_i^j\in\sigma(E_i)$, $\xi_{ij}^1$ is the top vector, and $\gamma_{ij}$ is the length. $i\in\{1,2,...,r\}$, $j\in\{1,2,...,p_i\}$.
\begin{lemma}\label{chara}
The eigenspace of $\Phi_s$ about $\theta_i^j$ is $M(\theta_i^j|\Phi_s)=V(\theta_i^j)$, where $V(\theta_i^j)= \text{span}\{\eta_{ij}^1,\eta_{ij}^2,...,\eta_{ij}^{\beta_{ij}}\}$, $\eta_{ij}^1=v_i^1\otimes\xi_{ij}^1$, $\eta_{ij}^2=v_i^2\otimes\xi_{ij}^1+v_i^1\otimes\xi_{ij}^2,...,\eta_{ij}^{\beta_{ij}}=v_i^{\beta_{ij}}\otimes\xi_{ij}^1+v_i^{\beta_{ij}-1}\otimes\xi_{ij}^2+...+v_i^1\otimes\xi_{ij}^{\beta_{ij}}$, $\beta_{ij}=\text{min}\{\alpha_i,\gamma_{ij}\}$, $i=1,2,...,r$ and $j=1,2,...,p_i$.
Specially, if $\theta_{i_1}^{j_1}=\theta_{i_2}^{j_2}=...=\theta_{i_l}^{j_l}\triangleq\theta_i^j$, where $i_k\in\{1,...,r\},j_k\in\{1,...,p_i\}$ for $k=1,...,l$, $l>1$, $M(\theta_i^j|\Phi_s)=\odot_{k=1}^l{V(\theta_{i_k}^{j_k})}$.
\end{lemma}

\begin{remark}
Lemmas \ref{lem1}-\ref{chara} are the sampled-data version of the results of Theorem 1 in \cite{hao2019new}.
In order to verify the networked sampled-data system based on the PBH rank condition, it is necessary to calculate all the eigenvalues and the whole eigenspace of the state matrix $\Phi_s$.
Lemmas \ref{lem1}-\ref{chara} provide a decomposition method to obtain the eigenvalues by lower-dimensional subsystems (represented by $E_i$), and to calculate the eigenspace by direct representation of Kronecher products of vectors with lower dimensions (i.e., $v_i^k$ and $\xi_{ij}^k$).
In this way, the curse of dimensionality can be avoided because there is no need to perform matrix operations directly on high-dimensional $\Phi_s$.
\end{remark}

Based on the above analysis, a necessary and sufficient controllability criterion for the networked sampled-data system
(\ref{MIMO_sampled_1},\ref{MIMO_sampled_2}) is given as Theorem \ref{thm-main}.

\begin{lemma}\label{tar_ref}\cite{hao2022target} \
The uncontrollable subspace of linear system $(A,B)$ is $\mathcal{M}(A,B)=span\{v^j_i|v^j_iB=0,...,v^1_iB=0,1\leq{i}\leq{r},1\leq{j}\leq\alpha_i\}$, where the eigenvalues of $A$ are $\lambda_1,...,\lambda_r$, and the left Jordan chain of $\lambda_i$ is $v^1_i,...,v^{\alpha_i}_i$.
\end{lemma}

\begin{lemma}\label{dlti}
A discrete-time LTI system $(A,B)$ is controllable if and only if $\forall\lambda\in\sigma(A)$, $\lambda\neq{0}$, and the corresponding eigenvector $\xi$ satisfy $\xi{B}\neq{0}$, where $A\in\mathbb{R}^{n\times{n}}$.
\end{lemma}

\begin{proof}
By the definition of state controllability, $(A,B)$ is controllable if and only if: $\exists k,u(\cdot)$ and $\forall x(0)$ that
\begin{equation}\label{eq1}
    0=x(kh)=A^kx(0)+\sum_{\tau=0}^{k-1}A^{k-\tau-1}Bu(\tau h).
\end{equation}
We only need to consider $n$-step. Denote $[B,AB,...,$ $A^{n-1}B]=\mathcal{C}$. The necessary and sufficient condition for (\ref{eq1}) to have a solution is $rank(\mathcal{C})=rank([\mathcal{C},A^n])$, which means $\mathcal{N}^\top(\mathcal{C})=\mathcal{N}^\top([\mathcal{C},A^n])$.\par
If $\exists\lambda\in\sigma(A)$,
$\lambda\neq{0}$, and $\xi{B}=0$, it is obvious that $\xi\in\mathcal{N}^\top(\mathcal{C})$, but $\xi{A}^n\neq{0}$, and thus $(A,B)$ is uncontrollable.
In turn, assume that $(A,B)$ is uncontrollable.
Then $\exists\xi\in\mathcal{N}^\top(\mathcal{C})$, and $\xi{A}^n\neq{0}$.
According to Lemma \ref{tar_ref}, $\mathcal{N}^\top(\mathcal{C})=\mathcal{M}(A,B)$, which can be spanned by generalized eigenvectors of $A$.
Therefore $\exists\lambda_i\in\sigma(A)$, $\lambda_i\neq{0}$ and its eigenvector $v_i^1\in\mathcal{M}(A,B)$ such that $v_i^1B=0$.
\end{proof}

\begin{theorem}
\label{thm-main}
The networked sampled-data system (\ref{MIMO_sampled_1},\ref{MIMO_sampled_2}) is controllable if and only if $\forall\theta_i^j\in\sigma(\Phi_s)$, $\theta_i^j\neq{0}$, and $\forall\eta\in{M(\theta_i^j|\Phi_s)}$ and $\eta\neq{0}$, $\eta(\Delta\otimes\mathcal{B}(h))\neq{0}$, where $M(\theta_i^j|\Phi_s)$ can be calculated by Lemma \ref{chara}, $i=1,...,r,j=1,...,p_i$.\par
\end{theorem}

\begin{proof}
Based on Lemma \ref{dlti}, a DLTI system is controllable if and only if the input matrix left multiplied by any eigenvector related to any nonzero eigenvalue of the state matrix is nonzero.
By Lemmas \ref{lem1}-\ref{chara}, $M(\theta_i^j|\Phi_s)$ with $i=1,...,r$ and $j=1,...,p_i$ is the whole eigenspace of the state matrix $\Phi_s$, thus Theorem \ref{thm-main} holds.
\end{proof}

Theorem \ref{thm-main} reveals that the controllability of a networked sampled-data system can be inferred from lower-dimensional operations by matrices and vectors related to some subsystems.
The following example demonstrates the specific verification process by Theorem \ref{thm-main}.

\begin{example}
\label{exm_thm1}
Consider a networked sampled-data system of three identical node systems in a chain structure, where $w_{21}=w_{32}=1,\delta_1=1,h=0.1$, and $B=I_2$,
$$\begin{aligned}
&A=\begin{bmatrix}
1 & 0\\
1 & 1
\end{bmatrix},
HC=\begin{bmatrix}
1 & 0\\
0 & 0
\end{bmatrix},
e^{Ah}=\begin{bmatrix}
    1.1052 & 0 \\
    0.1105 & 1.1052
\end{bmatrix},\\
&\mathcal{H}(h)=\begin{bmatrix}
    0.1052 & 0 \\
    0.0053 & 0
\end{bmatrix},
\mathcal{B}(h)=\begin{bmatrix}
    0.1052 & 0 \\
    0.0053 & 0.1052
\end{bmatrix}.
\end{aligned}$$
It can be computed that the associated generalized left eigenvectors of $W$ with respect to $\lambda_1=0$ are $v_1^1=[1,0,0]$ and $v_1^2=[0,1,0]$. Then
$$E_1=e^{Ah}+\lambda_1\mathcal{H}(h)
=\left[\begin{array}{cc}
    1.1052 & 0 \\
    0.1105 & 1.1052
\end{array}\right].$$
The eigenvalue of $E_1$ is $\theta_1^1=1.1052$, with the generalized left Jordan chain of $E_1$ about $\mathcal{H}(h)$ corresponding to $\theta_1^1$ being $\xi_{11}^1=[1, 0]$, $\xi_{11}^2=[0, -0.9520]$.
Then $\eta_{11}^1=v_1^1\otimes\xi_{11}^1=[1, 0, 0, 0, 0, 0]$ and $\eta_{11}^2=v_1^2\otimes\xi_{11}^1+v_1^1\otimes\xi_{11}^2=[0,-0.9520,1,0,0,0]$.\par
One can verify that for any $a_1,a_2\in\mathbb{C}$ and $[a_1,a_2]\neq{0}$,
$$\begin{aligned}
&(a_1\eta_{11}^1+a_2\eta_{11}^2)(\Delta\otimes\mathcal{B}(h))\\
=&[a_1,-0.9520a_2,a_2,0,0,0]
\left[\begin{array}{cccc}
0.1052 & 0 & \cdots & 0\\
0.0053 & 0.1052 &  & \\
\vdots &  & \ddots & \vdots\\
0 & \cdots &  & 0
\end{array}\right]\\
\neq&{0}.
\end{aligned}$$
Therefore, this networked sampled-data system is controllable.
\end{example}

\begin{remark}\label{nonsin}
Note that Theorem \ref{thm-main} does not need to verify eigenvalue zero, which is different from Theorem 2 in \cite{hao2019new}.
The first \textbf{Note} on Page 102 in \cite{hespanha2018linear} also says that a DLTI system's reachable subspace and controllable subspace could not be equivalent when the state matrix is singular.
Since the PBH rank condition is derived based on the image space of the controllability matrix (which is equal to the reachable subspace and strictly included in the controllable subspace), it is sufficient but unnecessary.
In this case, the controllable subspace can still be $\mathbb{R}^{Nn}$ even if the PBH rank condition is not satisfied with eigenvalue zero.
On the contrary, by the last \textbf{Note} on Page 101 in \cite{hespanha2018linear}, the PBH rank condition is still necessary and sufficient if $\Phi_s$ is nonsingular, i.e.,
$E_i$ has no zero eigenvalues for every $i=1,2,...,r$.
\end{remark}

The following is a counterexample showing that the PBH rank condition is sufficient but unnecessary.

\begin{example}
Consider a networked sampled-data system consisting of three identical nodes in a chain structure with self-rings, where $w_{11}=w_{22}=w_{33}=-1,w_{12}=w_{23}=1,\delta_1=1,h=0.1$, and
$$A=\left[\begin{array}{cc}
    0 & 1 \\
    0 & 0\end{array}\right],
    B=\left[\begin{array}{c}
    1 \\
    0\end{array}\right],
    HC=\left[\begin{array}{cc}
    10 & 0 \\
    0 & 10\end{array}\right].$$
According to the definition of controllability of DLTI systems, $(\Phi_s,\Psi_s)$ is controllable if and only if $\forall{X(0)}\in\mathbb{R}^{6}$, there exist $U(0),...,U(kh)\in\mathbb{R}^{3}$, and finite $k$ such that
    \begin{equation}\label{c_exm}
    0=\Phi_s^kX(0)+\sum_{\tau=0}^{k-1}\Phi_s^{k-\tau-1}\Psi_sU(\tau{h}).
    \end{equation}
    By simple calculation, it is easy to find that $\Phi_s$ is a nilpotent matrix.
    Therefore, if one takes $U(\cdot)=0$ and $k\geq6$, equation (\ref{c_exm}) holds, which indicates that the system is controllable.
    However, $\sigma(W)=\{-1\}$, and $E_1=e^{Ah}-\mathcal{H}(h)$.
    It can be calculated that $\sigma(E_1)=\{0\}$.
    By Definitions \ref{def1}-\ref{def2}, one has $v_1^1=[0,0,1],v_1^2=[0,1,0],\xi_{11}^1=[0,1],\xi_{11}^2=[-20,*]$.
    Therefore, the eigenvectors of $\Phi_s$ corresponding to the eigenvalue $0$ are $\eta_{11}^1=v_1^1\otimes\xi_{11}^1=[0,0,0,0,0,1]$, and $\eta_{11}^2=v_1^1\otimes\xi_{11}^2+v_1^2\otimes\xi_{11}^1=[0,0,0,1,-20,*]$.
    It is obvious that $\eta_{11}^1(\Delta\otimes\mathcal{B}(h))=0$ and $\eta_{11}^2(\Delta\otimes\mathcal{B}(h))=0$.
    If the PBH rank condition is necessary, it means that the system is uncontrollable, which leads to a contradiction.
\end{example}

\begin{remark}
In \cite{hao2019new}, it was revealed that the controllability of an LTI system is affected by network topology, node dynamics, external control inputs, and inner couplings.
Theorem \ref{thm-main} further develops the condition for the corresponding sampled-data system, where the subsystem $(A+\lambda_iH,B)$ in \cite{hao2019new} is replaced by $(e^{Ah}+\lambda_i\mathcal{H}(h),\mathcal{B}(h))$.
Therefore, it can be explicitly found that the sampling period $h$ also has effects on the controllability of the networked sampled-data system and is coupled with other network-related factors.
\end{remark}

When the network topology matrix is diagonalizable, an easy-verified condition is given below.

\begin{theorem}\label{thm2}
Assume that $W$ is diagonalizable, the networked sampled-data system (\ref{MIMO_sampled_1},\ref{MIMO_sampled_2}) is controllable if the following conditions hold simultaneously:\par
(1) $(W,\Delta)$ is controllable;\par
(2) $(E_i,\mathcal{B}(h))$ is controllable for every $i=1,2,...,N$;\par
(3) If $\theta\in\mathbb{C}$ is a common eigenvalue of $E_{k_1},E_{k_2},...,E_{k_q}$, $1<q\leq{N}$, then $(v_{k_1}\otimes\xi_{k_1}+v_{k_2}\otimes\xi_{k_2}+...+v_{k_q}\otimes\xi_{k_q})(\Delta\otimes{\mathcal{B}(h)})\neq{0}$ for $\forall\xi_j\in{M(\theta|E_j)}$, with $j=k_1,...,k_q$, and $[\xi_{k_1},...,\xi_{k_q}]\neq{0}$.
\end{theorem}

\begin{proof}
According to the PBH rank condition for DLTI systems, system (\ref{MIMO_sampled_1},\ref{MIMO_sampled_2}) is controllable if $[sI_{Nn}-\Phi_s,\Psi_s]$ is of full row rank for $\forall{s}\in\mathbb{C}$.
Since all of the Jordan blocks of $W$ are one-dimensional, denote the $N$ linearly independent eigenvectors of $W$ by $v_1,v_2,...,v_N$, and
\begin{equation}
    T=[v_1^\top,...,v_N^\top]^\top, \  T\Delta=[d_1^\top,...,d_N^\top]^\top,\nonumber
\end{equation}
where $v_i\in\mathbb{C}^{1\times{N}}$, $d_i\in\mathbb{C}^{1\times{N}}$, $i=1,...,N$.
Then make a similarity transformation:
$$\begin{aligned}
&[sI_{Nn}-\Phi_s,\Psi_s]\\
=&[sI_{Nn}-I_N\otimes{e^{Ah}}-W\otimes{\mathcal{H}(h)},\Delta\otimes{\mathcal{B}(h)}]\\
=&(T^{-1}\otimes{I_n})[sI_{Nn}-\mathcal{F},(T\Delta)\otimes\mathcal{B}(h)]\left[ \begin{array}{cc}
T\otimes{I_n} & 0\\
0 & I_{Np}
\end{array}
\right ]\\
=&(T^{-1}\otimes{I_n})[sI_{Nn}-\mathcal{F},\mathcal{G}]\left[ \begin{array}{cc}
T\otimes{I_n} & 0\\
0 & I_{Np}
\end{array}
\right ],
\end{aligned}$$
where 
\begin{equation}
[sI_{Nn}-\mathcal{F},\mathcal{G}]=
\left[ {\begin{array}{c:c}
\begin{matrix}
sI_n-E_1 &  &  \\
 & \ddots &  \\
 &  & sI_n-E_N \\
\end{matrix}&
\begin{matrix}
G_1 \\
\vdots \\
G_N \\
\end{matrix}
\end{array}} \right]\nonumber
\end{equation}
\begin{equation}
G_i=d_i\otimes\mathcal{B}(h), \ i=1,2,...,r.\nonumber
\end{equation}
Since $T$ is invertible, the rank of $[sI_{Nn}-\Phi_s,\Psi_s]$ equals to the rank of
$[sI_{Nn}-\mathcal{F},\mathcal{G}]$.
Then the proof can be completed by showing: If $(\mathcal{F},\mathcal{G})$ is not controllable, it will lead to the contradiction of one of the conditions in Theorem \ref{thm2}.\par

Firstly, consider the case that $(E_k,G_k)$ is not controllable for some $k\in\{1,2,...,N\}$.
If $d_k={0}$, $v_k[\lambda_k{I_N}-W,\Delta]=[0,d_k]=0$, which implies that $(W,\Delta)$ is uncontrollable and condition (1) does not hold.
Then assume that $d_k\neq{0}$. 
If $rank([sI_n-E_k,d_k\otimes\mathcal{B}(h)])=rank([sI_n-E_k,\mathcal{B}(h)])<n$, condition (2) does not hold.\par

Finally, discuss the last possibility that $(E_k,G_k)$ is controllable for every $k=1,2,...,N$, but there exist some $\theta$ being a common eigenvalue of $E_{k_1},...,E_{k_q}$, $1<q\leq{N}$, making $[\theta{I_{Nn}}-\mathcal{F},\mathcal{G}]$ not of full row rank.
Then there exists some $\xi=[\xi_{k_1},...,\xi_{k_q}]$, $\xi_j\in{M}(\theta|E_j)$, $j=k_1,...,k_q$, satisfying
$$\begin{aligned}
&\xi\left[ {\begin{array}{c:c}
\begin{matrix}
\theta{I_n}-E_{k_1} &       & \\
         &\ddots & \\
         &       & \theta{I_n}-E_{k_q} \end{matrix}&
\begin{matrix}
d_{k_1}\otimes\mathcal{B}(h)\\
\vdots\\
d_{k_q}\otimes\mathcal{B}(h)
\end{matrix}\end{array}} \right]\\
=&\left[0,(\xi_{k_1}(d_{k_1}\otimes\mathcal{B}(h))+...+\xi_{k_q}(d_{k_q}\otimes\mathcal{B}(h))\right]\\
=&\left[0,(v_{k_1}\otimes\xi_{k_1}+...+v_{k_q}\otimes\xi_{k_q})(\Delta\otimes\mathcal{B}(h))\right]=0,
\end{aligned}$$
condition (3) is contradicted.
The proof is complete.
\end{proof}

\begin{remark}
To verify the state controllability of the networked sampled-data system, the computational complexity of the Kalman criterion or the PBH rank condition is $\mathcal{O}(N^4n^4)$.
However, if it is checked by Theorem \ref{thm-main}, the computational complexity is no more than $\mathcal{O}(N^4+n^4N+N^3n^3)$ (consistent with the analysis in Remark 1 of \cite{hao2019new}).
As for Theorem \ref{thm2}, it is more intuitively that the computational complexity of conditions (1)-(3) is no more than $\mathcal{O}(N^4),\mathcal{O}(n^4N),\mathcal{O}(n^3N^3)$, respectively.
Since each eigenvalue is either common or non-common, the actual computational complexity can be even lower.
In addition, the steps such as the eigenspace decomposition of state matrices and the verification of subsystems can be performed in parallel in practical applications.
Therefore, the proposed criteria can reduce the computational burden.
\end{remark}

Note that Theorem \ref{thm-main} can still verify the networked sampled-data system with a diagonalizable topology matrix.
However, in any case, it requires calculating the whole eigenspace of $Nn$-dimension, and multiplying input matrix $\Delta\otimes\mathcal{B}(h)\in\mathbb{R}^{Nn\times{Np}}$.
Therefore, there are still high-dimensional operations in the verification process.
By Theorem \ref{thm2}, $Nn$-dimensional vector operations are involved only when there are common eigenvalues (corresponding to condition (3)).
For non-common eigenvalues, the controllability can be checked just by conditions (1) and (2), which only involves vector operations of $N$-dimension and $n$-dimension.
The following example illustrates the utilization of Theorem \ref{thm2}.

\begin{example}
Consider a networked sampled-data system consisting of three identical nodes, where $A,B,\Delta,h$ are the same as that in Example \ref{exm_thm1}, so are $e^{Ah}$ and $\mathcal{B}(h)$, and $w_{12}=w_{21}=w_{23}=w_{32}=1,HC=I_2$.
It is easy to compute that $\sigma(W)=\{0,1.4142,-1.4142\}$, then
$$\begin{aligned}
&E_1=e^{Ah}=\left[\begin{array}{cc}
    1.1052 & 0 \\
    0.1105 & 1.1052
\end{array}\right], \\
&E_2=e^{Ah}+1.4142\mathcal{H}(h)=\left[\begin{array}{cc}
    1.2539 & 0 \\
    0.1181 & 1.2539
\end{array}\right], \\
&E_3=e^{Ah}-1.4142\mathcal{H}(h)=\left[\begin{array}{cc}
    0.9564 & 0 \\
    0.1030 & 0.9564
\end{array}\right].
\end{aligned}$$
Since $W$ is diagonalizable, the controllability of the networked sampled-data system can be tested by Theorem \ref{thm2}.\par
Obviously $(W,\Delta)$ is controllable.
Next inspect the controllability of $(E_1,\mathcal{B}(h))$, $(E_2,\mathcal{B}(h))$ and $(E_3,\mathcal{B}(h))$.
By simple calculation it shows that $\forall{s}\in\mathbb{C}$,
$$\begin{aligned}
 &rank([sI_2-E_1,\mathcal{B}(h)]) \\
= \ &rank(\left[\begin{array}{cccc}
    s-1.1052 & 0 & 0.1052 & 0 \\
    -0.1105 & s-1.1052 & 0.0053 & 0.1052
\end{array}\right])=2,\\
\end{aligned}
$$
$$\begin{aligned}
&rank([sI_2-E_2,\mathcal{B}(h)])\\
= \ &rank(\left[\begin{array}{cccc}
    s-1.2539 & 0 & 0.1052 & 0 \\
    -0.1181 & s-1.2539 & 0.0053 & 0.1052
\end{array}\right])=2,\\ 
&rank([sI_2-E_3,\mathcal{B}(h)])\\
= \ &rank(\left[\begin{array}{cccc}
    s-0.9564 & 0 & 0.1052 & 0 \\
    -0.1030 & s-0.9564 & 0.0053 & 0.1052
\end{array}\right])=2.
\end{aligned}$$\par
Since $E_1,E_2$ and $E_3$ share no common eigenvalues, condition (3) of Theorem \ref{thm2} is no need to check.
Therefore, the system is controllable according to Theorem \ref{thm2}.
\label{exm_ex_2}
\end{example}

\begin{remark}
For the networked system with an undirected (or bidirectional) topological structure, its $W$ is a real symmetric matrix and thus can be diagonalized.
Therefore, the conditions in Theorem \ref{thm2} can be applied to the undirected (or bidirectional) networked sampled-data systems.
In addition, many other topological structures can also be represented by diagonalizable $W$, such as cycles discussed in the next section.
\end{remark}

According to the analysis in Remark \ref{nonsin}, a necessary and sufficient controllability criterion can be derived as follows.

\begin{corollary}
\label{coro1}
Assume that $W$ is diagonalizable and $E_i$ is nonsingular for every $i=1,...,N$. 
The networked sampled-data system (\ref{MIMO_sampled_1},\ref{MIMO_sampled_2}) is controllable if and only if conditions (1), (2) and (3) in Theorem \ref{thm2} hold simultaneously.
\end{corollary}

\begin{proof}
The sufficiency part of the proof has been shown as the proof of Theorem \ref{thm2}.\par

Necessity: If $\Phi_s$ is nonsingular, system (\ref{MIMO_sampled_1},\ref{MIMO_sampled_2}) is controllable only if $[sI_{Nn}-\mathcal{F},\mathcal{G}]$ is of full row rank.\par

If $(W,\Delta)$ is uncontrollable, there exists some $\lambda_k\in\sigma(W)$, $k\in\{1,...,N\}$ and its corresponding eigenvector $v_k$, satisfying $v_k\Delta=d_k=0$.
When the geometric multiplicity of $\lambda_k$ is $1$, let $\theta_k\in\sigma(E_k)$, $E_k=e^{Ah}+\lambda_k\mathcal{H}(h)$, then $rank([\theta_k{I}_n-E_k,d_k\otimes\mathcal{B}(h)])= rank(\theta_k{I}_n-E_k)<n$, making $[sI_{Nn}-\mathcal{F},\mathcal{G}]$ not of full row rank when $s=\theta_k$.
When the geometric multiplicity of $\lambda_k$ is $q>1$,
it has $q$ linearly independent eigenvectors $v_{k_1},v_{k_2},...,v_{k_q}$, where $k_i\in\{1,2,...,N\}$, $i\in\{1,...,q\}$. 
Assume that there exists some $v_k=aT_k$, where $a=[a_1,a_2,...,a_q]\in\mathbb{C}^{1\times{q}},$ with $a\neq{0}$, $T_k=[v_{k_1}^\top,...,v_{k_q}^\top]^\top$, $D_k=[d_{k_1}^\top,...,d_{k_q}^\top]^\top$, such that $v_k\Delta=aT_k\Delta=aD_k=0$.
It is obvious that $E_{k_1}=...=E_{k_q}=e^{Ah}+\lambda_k\mathcal{H}(h)$.
Let $\theta_k\in\sigma(E_{k_1})$, $\xi_{k_1}\in{M}(\theta_k|E_{k_1})$, denote $\xi=a\otimes\xi_{k_1}$, one has
$$\begin{aligned}
&\xi\left[ {\begin{array}{c:c}
\begin{matrix}
\theta_k{I_n}-E_{k_1} &       & \\
         &\ddots & \\
         &       & \theta_k{I_n}-E_{k_q} \end{matrix}&
\begin{matrix}
d_{k_1}\otimes\mathcal{B}(h)\\
\vdots\\
d_{k_q}\otimes\mathcal{B}(h)
\end{matrix}
\end{array}} \right]\\
=&\left[0,(aD_k)\otimes(\xi_{k_1}\mathcal{B}(h))\right]\\
=&\left[0,0\otimes(\xi_{k_1}\mathcal{B}(h))\right]=0,
\end{aligned}$$
which also makes $[s{I_{Nn}}-\mathcal{F},\mathcal{G}]$ not of full row rank when $s=\theta_k$.\par

If $(E_k,\mathcal{B}(h))$ is uncontrollable for some $k\in\{1,2,...,r\}$, then $[\theta_kI_n-E_k,d_k\otimes\mathcal{B}(h)]$ and furthermore $[\theta_kI_{Nn}-\mathcal{F},\mathcal{G}]$ are not of full row rank.\par

Lastly, consider the case that $\theta$ is a common eigenvalue of $E_{k_1},...,E_{k_q}$, $1<q\leq{N}$ and $k_1,...,k_q\in\{1,2,...,r\}$. Assume that $\exists{\xi_j}\in{M(\theta|E_j)},j=k_1,...,k_q$, such that $(v_{k_1}\otimes\xi_{k_1}+v_{k_2}\otimes\xi_{k_2}+...+v_{k_q}\otimes\xi_{k_q})(\Delta\otimes{\mathcal{B}(h)})=0$.
Let $\eta=[\eta_1,...,\eta_N]$, where $\eta_j\in\mathbb{C}^{1\times{n}}$, $\eta_j=\xi_j$ if $j\in\{k_1,...,k_q\}$, otherwise $\eta_j=0$.
It follows that
$$\begin{aligned}
&\eta[\theta{I_{Nn}}-\mathcal{F},\mathcal{G}]\\
=&\left[0,(\xi_{k_1}(d_{k_1}\otimes\mathcal{B}(h))+...+\xi_{k_q}(d_{k_q}\otimes\mathcal{B}(h))\right]\\
=&\left[0,(v_{k_1}\otimes\xi_{k_1}+...+v_{k_q}\otimes\xi_{k_q})(\Delta\otimes\mathcal{B}(h))\right]=0.
\end{aligned}$$
Therefore, $[\theta{I_{Nn}}-\mathcal{F},\mathcal{G}]$ is not of full row rank, which completes
the proof.
\end{proof}

The following example intuitively shows the efficiency of Corollary \ref{coro1} to identify the uncontrollability of the networked sampled-data system.

\begin{example}
Reonsider the system in Example \ref{exm_ex_2}, but let
$$
H=\left[\begin{array}{cc}
    1 & 0 \\
    0 & 0
\end{array}\right].$$
Then one has
$$\begin{aligned}
&E_1=e^{Ah}=\left[\begin{array}{cc}
    1.1052 & 0 \\
    0.1105 & 1.1052
\end{array}\right], \\
&E_2=e^{Ah}+1.4142\mathcal{H}(h)=\left[\begin{array}{cc}
    1.2539 & 0 \\
    0.1181 & 1.1052
\end{array}\right], \\
&E_3=e^{Ah}-1.4142\mathcal{H}(h)=\left[\begin{array}{cc}
    0.9564 & 0 \\
    0.1030 & 1.1052
\end{array}\right].
\end{aligned}$$
Since $W$ is diagonalizable and $0$ is not an eigenvalue of $E_1$, $E_2$ or $E_3$, the controllability of the networked sampled-data system can be tested by Corollary \ref{coro1}.\par

Note that $1.1052$ is a common eigenvalue of $E_1,E_2$ and $E_3$, and $\xi_1=[1,0],\xi_2=[0.7939,0],\xi_3=[0.6922,1],v_1=[-1,0,1],v_2=[1,1.4142,1],v_3=[1,-1.4142,1]$.
Then take $\eta=v_1\otimes{a}_1\xi_1+v_2\otimes{a}_2\xi_2+v_3\otimes{a}_3\xi_3$.
When $a_1=0.7939,a_2=1,a_3=0,\eta(\Delta\otimes\mathcal{B}(h))=0$.
According to condition (3) of Corollary \ref{coro1}, the system is uncontrollable.
\end{example}

Some necessary conditions are listed as follows:

\begin{corollary}
\label{coro_add3}
Assume that $0\notin\sigma(E_i)$ for every $i=1,...,r$.
The networked system (\ref{MIMO_sampled_1},\ref{MIMO_sampled_2}) is controllable only if $(W,\Delta)$ and $(E_i,\mathcal{B}(h))$, $i=1,2,...,r$ are all controllable.
\end{corollary}

\begin{proof}
Corollary \ref{coro_add3} is induced from the necessary part of the proof of Corollary \ref{coro1}.
\end{proof}

\begin{corollary}
\label{coro2}
Assume that $W$ is singular, and $0\notin\sigma(E_i)$ for every $i=1,...,r$.
The networked system (\ref{MIMO_sampled_1},\ref{MIMO_sampled_2}) is controllable only if $(e^{Ah},\mathcal{B}(h))$ is controllable.
\end{corollary}

\begin{proof}
If $W$ is singular, there exists some $\lambda_k\in\sigma(W)$ such that $\lambda_k=0$. According to Corollary \ref{coro_add3}, $(E_k,\mathcal{B}(h))$ has to be controllable to ensure the controllability of system (\ref{MIMO_sampled_1},\ref{MIMO_sampled_2}).
Since $E_k=e^{Ah}$, the controllability of $(e^{Ah},\mathcal{B}(h))$ is necessary.
\end{proof}

Corollary \ref{coro2} can be used to effectively identify the uncontrollability of the whole networked sampled-data system by its node dynamics independently, which is demonstrated by the following example.
In this case, the uncontrollability of the whole networked system can be diagnosed even if the information about the precise network topology and inner couplings is unknown.
Note that Corollary \ref{coro2} is not a sufficient condition, since the controllability of the system (\ref{MIMO_sampled_1},\ref{MIMO_sampled_2}) is determined by multiple coupled factors. Even if $(e^{Ah},\mathcal{B}(h))$ is controllable, $W$ may have another nonzero eigenvalue $\lambda_i$, such that $(\theta_i^j{I_n}-E_i,\mathcal{B}(h))$ is not of full row rank for some $\theta_i^j\in\sigma(E_i)$, $i\in\{1,...,r\}$, $j\in\{1,...,p_i\}$.

\begin{example}
\label{exam2}
Consider a networked system consisting of three connected identical nodes with a chain structure, where $w_{21}$ and $w_{32}$ are nonzero,
$$A=\left[\begin{array}{cc}
    1 & 0 \\
    1 & 1
\end{array}\right],
B=\left[\begin{array}{cc}
    0 & 0 \\
    0 & 1
\end{array}\right],$$
and the sampling period is $h=0.1$. 
It is obvious that the eigenvalue of $W$ is $0$ and $E_1$ is nonsingular.
$$E_1=e^{Ah}=\left[\begin{array}{cc}
    1.1052 & 0 \\
    0.1105 & 1.1052
\end{array}\right],
\mathcal{B}(h)=\left[\begin{array}{cc}
    0 & 0 \\
    0 & 0.1052
\end{array}\right].$$
Then it can be calculated that
$$\begin{aligned}
&rank([sI_2-e^{Ah},\mathcal{B}(h)])\\
= \ &rank(\left[\begin{array}{cccc}
    s-1.1052 & 0 & 0 & 0 \\
    -0.1105 & s-1.1052 & 0 & 0.1052
\end{array}\right])=1<2
\end{aligned}$$
when $s=1.1052$,
by Corollary \ref{coro2}, the networked sampled-data system will be uncontrollable.
\end{example}

\begin{remark}\label{rem6}
The pathological sampling of $(A,B)$ will not inevitably result in the loss of controllability of the whole networked system, which will be shown in Example \ref{exam_new}.
Corollary \ref{coro2} essentially reveals that for the networked sampled-data system with singular $W$ and nonsingular $\Phi_s$, the pathological sampling of single node system $(A,B)$ cannot be eliminated in the network.
Especially, if $W$ only has zero eigenvalues (e.g., the tree structure, including chains and stars), the pathological sampling of $(A,B)$ will always result in the uncontrollability of the whole system.
\end{remark}

\begin{example}\label{exam_new}
Reconsider the networked sampled-data system in Fig.\ref{fig:sampled}, but extend the topology to three nodes in a cycle structure, where $\delta_1=1,w_{13}=w_{21}=w_{32}=1$, and
$$e^{Ah}=\left[\begin{array}{cc}
    -23.1407 & 0 \\
    0 & -23.1407
\end{array}\right],
\mathcal{B}(h)=\left[\begin{array}{cc}
    -12.0703 \\
    -12.0703
\end{array}\right].$$\par
The sampling is pathological about $A$, i.e., the controllability of the single node system is lost during the control sampling because
$$\begin{aligned}
&rank([sI_2-e^{Ah},\mathcal{B}(h)])\\
= \ &rank(\left[\begin{array}{ccc}
    s+23.1407 & 0 & -12.0703 \\
    0 & s+23.1407 & -12.0703
\end{array}\right])=1 
\end{aligned}$$
when $s=-23.1407$.
However, one can find that the whole networked sampled-data system is controllable by verifying Corollary \ref{coro1} as follows.\par
The eigenvalues of $W$ are $\lambda_1=1,\lambda_2=-0.5-0.866\rm{i},\lambda_3=-0.5+0.866\rm{i}$.
Then
$$\begin{aligned}
&\mathcal{H}(h)=\left[\begin{array}{cc}
    -12.0703 & 12.0703 \\
    -12.0703 & -12.0703
\end{array}\right],\\
&E_1{=}e^{Ah}{+}\lambda_1\mathcal{H}(h)=\left[\begin{array}{cc}
    -35.2110 & 12.0703 \\
    -12.0703 & -35.2110
\end{array}\right],\\
&E_2{=}e^{Ah}{+}\lambda_2\mathcal{H}(h)=\begin{bmatrix}
    -17.11+10.45\rm{i} & -6.035-10.45\rm{i} \\
    6.035+10.45\rm{i} & -17.11+10.45\rm{i}
\end{bmatrix},\\
&E_3{=}e^{Ah}{+}\lambda_3\mathcal{H}(h)=\begin{bmatrix}
    -17.11-10.45\rm{i} & -6.035+10.45\rm{i} \\
    6.035-10.45\rm{i} & -17.11-10.45\rm{i}
\end{bmatrix}.
\end{aligned}$$
It is easy to check that ($W,\Delta$) is controllable, and $rank([sI_2-E_1,\mathcal{B}(h)])$, $rank([sI_2-E_2,\mathcal{B}(h)])$, and $rank([sI_2-E_3,\mathcal{B}(h)])$ are all equal to $2$ for $\forall{s}\in\mathbb{C}$, and there are no common eigenvalues between $E_1,E_2$ and $E_3$.
\end{example}

\begin{corollary}\label{coro:dimen}
Assume that $0\notin\sigma(E_i)$ for every $i=1,...,r$. Assume that $(W,\Delta)$ is uncontrollable, and $\lambda_{k_1},...,\lambda_{k_q}$ are uncontrollable modes, corresponding to eigenvectors $v_{k_1}^1,...,v_{k_q}^1$, respectively.
Then the dimension of controllable subspace of system (\ref{MIMO_sampled_1},\ref{MIMO_sampled_2}) is no more than $Nn-\sum_{j=1}^q \mathcal{D}(E_{k_j})$.
\end{corollary}

\begin{proof}
For the state matrix $\Phi_s$ is nonsingular, the controllable subspace of $(\Phi_s,\Psi_s)$ is the column space of $\mathscr{C}=[\Phi_s,\Phi_s\Psi_s,...,\Phi_s^{Nn-1}]$, i.e., $\mathcal{R}(\mathscr{C})$.
According to Lemma \ref{chara}, if $v^1_{k_j}\Delta=0$, and $E_{k_j}$ has $\mathcal{D}(E_{k_j})$ independent eigenvectors, namely, $\xi^1_{k_j1},...,\xi^1_{k_j\mathcal{D}(E_{k_j})}$, it indicates that $\eta^1_{k_j1}=v^1_{k_j}\otimes\xi^1_{k_j1},...,\eta^1_{k_j\mathcal{D}(E_{k_j})}=v^1_{k_j}\otimes\xi^1_{k_j\mathcal{D}(E_{k_j})}$ are independent eigenvectors of $\Phi_s$, and $\eta^1_{k_jl}\Psi_s=0$ for every $l=1,...,\mathcal{D}(E_{k_j})$.
Therefore, $\eta^1_{k_jl}\in\mathcal{N}^\top(\mathscr{C})$ for every $l=1,...,\mathcal{D}(E_{k_j})$.
Consider all these eigenvectors of uncontrollable modes of $W$, i.e., $v_{k_1}^1,...,v_{k_q}^1$, one has
\begin{equation}
    \sum_{j=1}^q\mathcal{D}(E_{k_j})\leq dim(\mathcal{N}^\top(\mathscr{C}))=Nn-dim(\mathcal{R}(\mathscr{C})).\nonumber
\end{equation}
Then it comes to the result that the dimension of the controllable subspace of the system (\ref{MIMO_sampled_1},\ref{MIMO_sampled_2}) is no more than $Nn-\sum_{j=1}^q \mathcal{D}(E_{k_j})$.
\end{proof}

\section{Sampled-data systems with special topologies}
\label{sec:spe-topo}
In this section, the topology of directed trees and cycles are considered, with some easy-verified controllability conditions developed.

\subsection{Trees}
Consider a networked sampled-data system where the topology is a directed tree.
For example, a tree networked sampled-data system with six nodes is shown in Fig. \ref{fig:tree}, where nodes $1,2,5$ are under control.
The topology matrix can be written in lower triangular form with all the elements on the diagonal being zero.
Then it is obvious that all the eigenvalues of the topology matrix are zero.
Denote 
$$\sigma(\Phi_s)=\sigma(e^{Ah})=\{\sigma_1,...,\sigma_q\},1\leq{q}\leq{n}.$$
It is easy to find that $0\notin\Phi_s$, for $e^{Ah}$ is always nonsingular with arbitrary $A$ and $h$.
Therefore, the controllability conditions derived from the PBH rank condition become necessary and sufficient for systems with tree structures.\par
Note that a chain network is a tree network with only one leaf, and a star network is also a tree network, composed of one root node and multiple leaf nodes.
Examples of networked sampled-data systems with chain and star topology are shown in Fig. \ref{fig:chain} and Fig. \ref{fig:star}, respectively.
The controllability of networked sampled-data systems with these two types of topologies is analyzed as follows.

\begin{figure}[tb]
\centering
\begin{subfigure}[b]{0.5\textwidth}
  \centering
  \includegraphics[width=.55\linewidth]{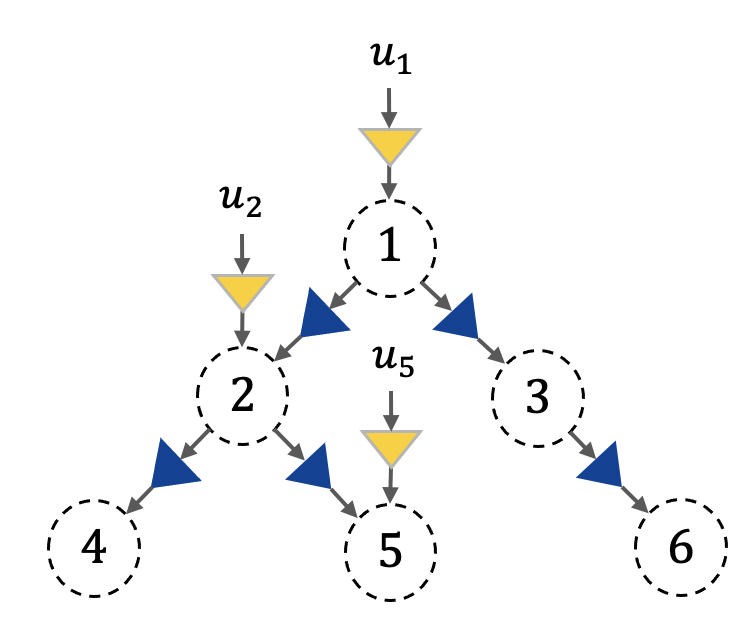}
  \caption{}
  \label{fig:tree}
\end{subfigure}
\quad
\begin{subfigure}[b]{.23\textwidth}
  \centering
  \includegraphics[width=.9\linewidth]{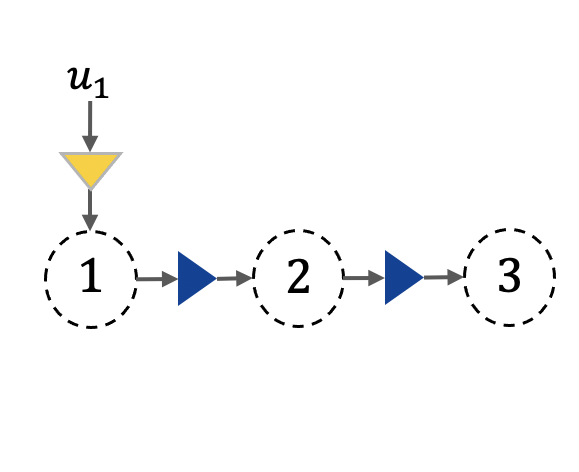}
  \caption{}
  \label{fig:chain}
\end{subfigure}
\begin{subfigure}[b]{.2\textwidth}
  \centering
  \includegraphics[width=.85\linewidth]{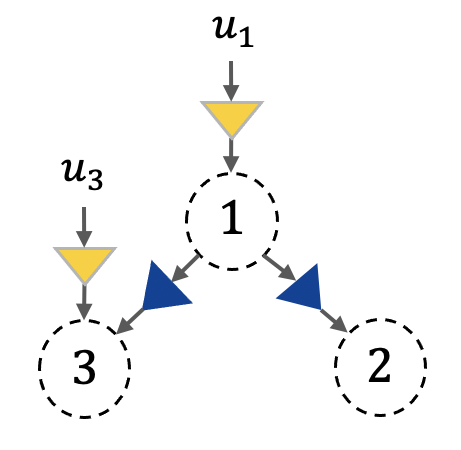} 
  \caption{}
  \label{fig:star}
\end{subfigure}
\caption{Networked sampled-data systems with (a) tree topology, (b) chain topology, (c) star topology.}
\label{fig:tree_topology}
\end{figure}

\subsubsection{Chains}
Consider a networked sampled-data system with a directed chain topology, where the topology matrix is in form of:
$$W_{chain}=\left[\begin{array}{cccc}
    0 &  & \cdots & 0\\
    w_{2,1} & 0 &  & \\
    \vdots & \ddots &\ddots & \vdots \\
    0 & \cdots & w_{N,N-1} & 0
\end{array}\right],$$
where $w_{i,i-1}$ is the weight of edge $\{i-1,i\}$, and $w_{i,i-1}\neq{0}$.
It is straightforward that the network topology is controllable if the first node is under control, i.e., $(W_{chain},\Delta)$ is controllable if $\Delta=\Delta_1$,
\begin{equation}
    \Delta_1=diag\{1,0,...,0\}.\nonumber
\end{equation}

Recall the notions of the Jordan chain and the generalized Jordan chain.
Let $v_i^1,v_i^2,...,v_i^N$ be a Jordan chain of $W_{chain}$ about $0$, and $\xi_i^1,...,\xi_i^{\gamma_i}$ be a generalized Jordan chain of $e^{Ah}$ about $\mathcal{H}(h)$ corresponding to $\sigma_i$, $i=1,...,q$.
Then an easier-to-verify controllability condition can be obtained for networked sampled-data systems with chain topology $(W_{chain},\Delta_1)$.

\begin{corollary}
The networked sampled-data system with chain topology $(W_{chain},\Delta_1)$ is controllable if and only if
$\eta\mathcal{B}(h)\neq{0}$, for $\forall\eta\in{V}_i$ and $\eta\neq{0}$, where $V_i=span\{\xi_i^1,...,\xi_i^{\beta_i}\}$, $\beta_i=min\{N,\gamma_i\}$, $i=1,...,q$.
Specially, if $\sigma_{i_1}=...=\sigma_{i_l}$, $l>1$, $V_i=span\{\xi_{i_1}^1,...,\xi_{i_1}^{\beta_{i_1}},...,\xi_{i_l}^1,...,\xi_{i_l}^{\beta_{i_l}}\}$. 
\end{corollary}

\begin{proof}
It can be calculated that the left Jordan chain of $W_{chain}$ with respect to $0$ is $v_1^1=e_1,v_1^2=(1/w_{21})e_2,...,v_1^N=(1/\prod_{k=2}^N w_{k,k-1})e_N$.
And the generalized Jordan chain of $e^{Ah}$ about $\mathcal{H}(h)$ corresponding to $\sigma_i$ is $\xi_i^1,...,\xi_i^{\gamma_i}$, $i=1,...,q$.
Similar to the proof of Theorem \ref{thm-main}, the eigenvectors of $\Phi_s$ about $\sigma_i$ are $\eta_i^1=v_1^1\otimes\xi_i^1,...,\eta_i^{\beta_i}=v_1^1\otimes\xi_i^{\beta_i}+...+v_1^{\beta_i}\otimes\xi_i^1$, where $\beta_i=min\{N,\gamma_i\}$.
According to the PBH rank condition, the networked sampled-data system is controllable if and only if $\forall{i}\in\{1,...,q\}$, $\forall\eta\in{span}\{\eta_i^1,...,\eta_i^{\beta_i}\}$ and $\eta\neq{0}$, $\eta(\Delta_1\otimes\mathcal{B}(h))\neq{0}$.
Since $\Delta_1=diag\{1,0,...,0\}$, the above condition is equivalent to that $\forall\eta\in
{span}\{\xi_i^1,...,\xi_i^{\beta_i}\}$ and $\eta\neq{0}$, $\eta\mathcal{B}(h)\neq{0}$.
Specially, if the geometric multiplicity of some eigenvalue of $e^{Ah}$ is greater than $1$, the test needs to consider all the general Jordan chains corresponding to this eigenvalue.
\end{proof}

\subsubsection{Stars}
Consider a networked sampled-data system with a directed star topology, where the topology matrix is in the form of:
$$W_{star}=\left[\begin{array}{cccc}
    0 & 0 & \cdots & 0\\
    w_{2,1} & 0 &  & \\
    \vdots & \ddots &\ddots & \vdots \\
    w_{N,1} & \cdots & 0 & 0
\end{array}\right],$$
where $w_{i,1}$ denotes the weight of edge $\{1,i\}$, and $w_{i,1}\neq{0}$.
It can be simply found that the network topology of a star network is controllable if the root and at least $N-2$ leaf nodes are under control.
Without loss of generality, let node $2$ be not under control, i.e., $(W_{star},\Delta)$ is controllable if $\Delta=\Delta_2,$
$$\Delta_2=diag\{1,0,1,...,1\}.$$
Then another easy-to-verify controllability condition can be obtained for networked sampled-data systems with star topology $(W_{star},\Delta_2)$.

\begin{corollary}\label{star}
The networked sampled-data system with chain topology $(W_{star},\Delta_2)$ is controllale if and only if $\eta\mathcal{B}(h)\neq{0}$, for $\forall\eta\in{V}_i$ and $\eta\neq{0}$, where $V_i=span\{\xi_i^1,...,\xi_i^{\beta_i}\}$, $\beta_i=min\{2,\gamma_i\}$, $i=1,...,q$.
Specially, if $\sigma_{i_1}=...=\sigma_{i_l}$, $l>1$, $V_i=span\{\xi_{i_1}^1,...,\xi_{i_1}^{\beta_{i_1}},...,\xi_{i_l}^1,...,\xi_{i_l}^{\beta_{i_l}}\}$. 
\end{corollary}

\begin{proof}
The generalized left eigenvectors of $W_{star}$ with respect to $0$ are $v_1^1=e_1,v_1^2=(1/w_{2,1})e_2,v_2^1=e_3-(w_{3,1}/w_{2,1})e_2,...,v_{N-1}^1=e_N-(w_{N,1}/w_{2,1})e_2$.
For $i\in\{1,...,q\}$, let the generalized left Jordan chain of $e^{Ah}$ about $\mathcal{H}(h)$ corresponding to $\sigma_i$ be $\xi_i^1,...,\xi_i^{\gamma_i}$.\par

If $\gamma_i=1$, the eigenvectors of $\Phi_s$ about $\sigma_i$ are $\eta_i^1=v_1^1\otimes\xi_i^1,\eta_i^2=v_2^1\otimes\xi_i^1,...,\eta_i^{N-1}=v_{N-1}^1\otimes\xi_i^1$.
According to the PBH rank condition, the networked sampled-data system is controllable if and only if $\forall{i}\in\{1,...,q\}$, $\forall\eta\in{span}\{\eta_i^1,...,\eta_i^{N-1}\}$ and $\eta\neq{0}$, $\eta(\Delta_2\otimes\mathcal{B}(h))\neq{0}$.
Since $\Delta_2=diag\{1,0,1,...,1\}$, the above condition is equivalent to that $\xi_i^1\otimes\mathcal{B}(h)\neq{0}$.\par

If $\gamma_i>1$, the eigenvectors of $\Phi_s$ about $\sigma_i$ are $\eta_i^1=v_1^1\otimes\xi_i^1,...,\eta_i^{N-1}=v_{N-1}^1\otimes\xi_i^1,\eta_i^N=v_1^2\otimes\xi_i^1+v_1^1\otimes\xi_i^2$.
According to the PBH rank condition, the networked sampled-data system is controllable if and only if $\forall{i}\in\{1,...,q\}$, $\forall\eta\in{span}\{\eta_i^1,...,\eta_i^N\}$ and $\eta\neq{0}$, $\eta(\Delta_2\otimes\mathcal{B}(h))\neq{0}$.
Since $\Delta_2=diag\{1,0,1,...,1\}$, the above condition is equivalent to that $\forall\eta\in{span}\{\xi_i^1,\xi_i^2\}$ and $\eta\neq{0}$, $\eta\mathcal{B}(h)\neq{0}$.\par
Specially, if the geometric multiplicity of some eigenvalue of $e^{Ah}$ is greater than $1$, the test needs to consider all the general Jordan chains corresponding to this eigenvalue.
In summary, the result in Corollary \ref{star} holds.
\end{proof}

\begin{corollary}\label{star_co}
Assume that $\gamma_i=1$ for every $i=1,...,q$.
The networked sampled-data system with star topology $(W_{star},\Delta_2)$ is controllable if and only if $(A,B)$ is controllable and $h$ is non-pathological about $A$.
\end{corollary}

\begin{proof}
Corollary \ref{star_co} can be proved based on the proof of Corollary \ref{star} and the non-pathological sampling condition of single systems.
\end{proof}

\begin{figure}[tb]
\centering
\includegraphics[width=.5\linewidth]{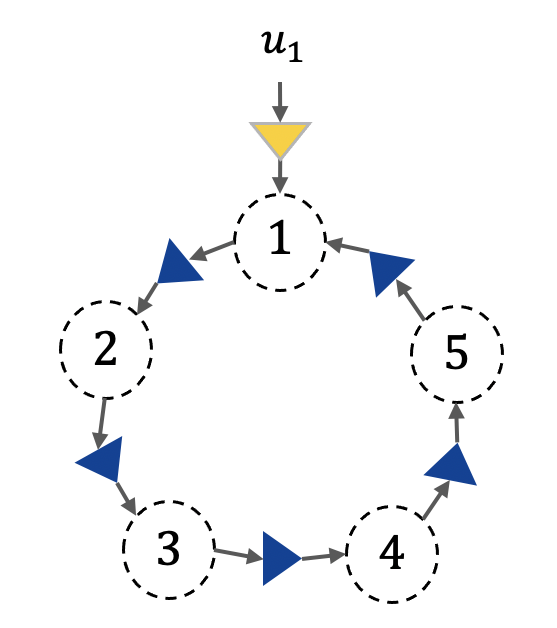} 
\caption{Networked sampled-data system with cycle topology.}
\label{fig:circle}
\end{figure}

\subsection{Cycles}
Consider a networked sampled-data system with a directed cycle topology, where the topology matrix is in form of:\par
$$W_{cycle}=\left[\begin{array}{cccc}
    0 & 0 & \cdots & w_{1,N}\\
    w_{2,1} & 0 & \cdots & 0\\
    \vdots & \ddots & \ddots & \vdots \\
    0 & \cdots & w_{N,N-1} & 0
\end{array}\right],$$
with $w_{1,N}\neq{0}$ and $w_{i,i-1}\neq{0}$.
It can be learned that only one external input added to an arbitrary node is enough for the controllability of the network topology.
Without loss of generality, assume that the external input is added to the first node, i.e., $(W_{cycle},\Delta)$ is controllable if $\Delta=\Delta_1$.
For instance, a cycle networked sampled-data system with five nodes is shown in Fig.\ref{fig:circle}, where only node $1$ is under control.\par

The characteristic polynomial of $W_{cycle}$ is
$$|\lambda{I}_N-W_{cycle}|=\lambda^N-w_{1,N}\prod_{i=2}^N w_{i,i-1}.$$
Denote $w_{1,N}\prod_{i=2}^N w_{i,i-1}$ by $\bar{w}$. 
Since $\bar{w}$ is a real number, the solutions of $|\lambda{I}_N-W_{cycle}|=0$ are:
\begin{equation}\label{eig_cir}
\lambda_i=\left\{
        \begin{array}{l}
        \sqrt[N]{\bar{w}}e^{\rm{i}\theta_i},\theta_i=\frac{2i\pi}{N},\bar{w}>0,\\
        \sqrt[N]{-\bar{w}}e^{\rm{i}\theta_i},\theta_i=\frac{(2i-1)\pi}{N},\bar{w}<0,\\
        \end{array}
        i=1,...,N.
\right.
\end{equation}

Therefore, $N$ eigenvalues of $W_{cycle}$ are different from each other, i.e., all the Jordan blocks of $W_{cycle}$ are one-dimensional.

\begin{corollary}\label{thm:circle}
The networked sampled-data system with cycle topology $(W_cycle,\Delta_1)$ is controllable if conditions (2) and (3) in Theorem \ref{thm2} hold simultaneously.
\end{corollary}

\begin{proof}
Since $W_{cycle}$ is diagonalizable, Corollary \ref{thm:circle} can be derived from the results in Corollary \ref{coro1}.
\end{proof}

\section{Networked sampled-data systems with special dynamics}\label{sec:spe-dyna}
In this section, two types of special node dynamics are considered: one-dimensional dynamics and self-loop-state dynamics. 

\subsection{One-dimensional dynamics}
Consider a networked system with one-dimensional dynamics, whose state matrix and input matrix are reduced to scalars denoted by $a\neq{0}$ and $b$, respectively.
Apparently, $HC$ is also a scalar, denoted by $c$,
and $a,b,c\in\mathbb{R}$.
The model of networked systems with one-dimensional mode dynamics is a special variant of the networked system introduced in Section \ref{sec:pre}.C.
In the continuous-time system (\ref{con:multi_linear}),
\begin{equation}\label{con:linear1}
    \Phi=aI_N+cW,\Psi=b\Delta,
\end{equation}
and in the related sampled-data system (\ref{MIMO_sampled_1}),
\begin{equation}\label{sampled1}
    \Phi_s=e^{ah}I_N+\frac{c}{a}(e^{ah}-1)W, \Psi_s=\frac{b}{a}(e^{ah}-1)\Delta.
\end{equation}

A controllability condition for the system (\ref{MIMO_sampled_1},\ref{sampled1}) can be derived based on Theorem \ref{thm2} as follows.

\begin{corollary}\label{thm-1}
The networked sampled-data system with one-dimensional node dynamics (\ref{MIMO_sampled_1},\ref{sampled1}) is controllable if $(W,\Delta)$ is controllable and $b\neq{0}$, $c\neq{0}$.
\end{corollary}

\begin{proof}
Similar to the derivation of Theorem \ref{thm2}, system (\ref{MIMO_sampled_1},\ref{sampled1}) is controllable if
$$\begin{aligned}
&[sI_N-e^{ah}I_N-\frac{c}{a}(e^{ah}-1)W,\frac{b}{a}(e^{ah}-1)\Delta]\\
=&T^{-1}[(s-e^{ah})I_N-\frac{c}{a}(e^{ah}-1)J,bT\Delta]\\
&\left[ \begin{array}{cc}
T & 0\\
0 & \frac{1}{a}(e^{ah}-1)I_N
\end{array}
\right ]\\
=&T^{-1}
(sI_N-\mathcal{F},\mathcal{G} )
\left[ \begin{array}{cc}
T & 0\\
0 & \frac{1}{a}(e^{ah}-1)I_N
\end{array}
\right ]
\end{aligned}$$
is of full row rank, where
$\mathcal{F}=diag\{F_1,...F_r\}, \mathcal{G}=[G_1^\top,...G_r^\top]^\top,$ and for $i=1,2,...,r$,
$$F_i=\left[\begin{matrix}
E_i & -\frac{c}{a}(e^{ah}-1) & & \\
         &\ddots &\ddots & \\
         &       & E_i & -\frac{c}{a}(e^{ah}-1)\\
         &       &          & E_i \end{matrix}\right]_{(\alpha_i\times{\alpha_i})}$$
$$G_i=\left[ \begin{array}{c}
bd_i^{\alpha_i}\\
\vdots\\
bd_i^1
\end{array}\right], 
E_i=e^{ah}+\lambda_i{\frac{c}{a}(e^{ah}-1)}\in\mathbb{C}.$$
\par

Then prove by contradiction.
Assume that system (\ref{MIMO_sampled_1},\ref{sampled1}) is uncontrollable.
Since $E_1,...,E_r$ are scalars, $\sigma(\Phi_s)=\{E_1,...,E_r\}$.
Consider $E_k$, $k\in\{1,...,r\}$.
If $[E_kI_N-\mathcal{F},\mathcal{G}]$ is not of full rank, there exists some nonzero $\eta=[\eta_1,...,\eta_r]$, such that $\eta\mathcal{G}=0$,
where $\eta_i=a_ie_{\alpha_i}\in\mathbb{C}^{1\times{\alpha_i}}$, $a_i$ is an arbitrary nonzero complex number if $E_i=E_k$, otherwise $a_i=0$, $i=1,...,r$.
If $a,b,c\neq{0}$, $\forall{i,j}\in\{1,...,r\}$, $i\neq{j}$, $E_i=E_j$ if and only if $\lambda_i=\lambda_j$.
Therefore, $\eta\mathcal{G}=0$ is equivalent to $\sum_{i=1}^r a_id_i^1=(\sum_{i=1}^r a_iv_i^1)\Delta=0$, indicating that $(W,\Delta)$ is uncontrollable.
\end{proof}

\begin{remark}
Note that $c\neq{0}$ in Corollary \ref{thm-1} is not necessary.
Assume that $c=0$,
let $\Delta=I_N$, system (\ref{MIMO_sampled_1},\ref{sampled1}) is still controllable.
It shows that when the transmission channels are cut off, external inputs must be added to each node to maintain the controllability of the overall system.
This also reveals the enhanced effects of interactions between nodes on the controllability of networked sampled-data systems.
\end{remark}

Note that $e^{ah}$ is always a nonzero scalar.
It seems that no effects of sampling are reflected in the controllability of the system (\ref{MIMO_sampled_1},\ref{sampled1}).
In fact, when the node dynamics are one-dimensional, the periodic sampling does not influence the controllability of the networked system, which is shown in Corollary \ref{thm:1state}.
Therefore, the controllability of the networked sampled-data system can be directly inferred from its original continuous-time system.

\begin{corollary}
\label{thm:1state}
The networked sampled-data system with one-dimensional node dynamics (\ref{MIMO_sampled_1},\ref{sampled1}) is controllable if its original continuous-time system (\ref{con:multi_linear},\ref{con:linear1}) is controllable.
\end{corollary}

\begin{proof}
Proof by contradiction.
If system (\ref{MIMO_sampled_1},\ref{sampled1}) is uncontrollable, there exists some $\xi\in\mathbb{C}^{1\times{N}}$ and $s_1\in\mathbb{C}$, such that $(s_1-e^{ah})\xi=\frac{c}{a}(e^{ah}-1)\xi{W}$ and $\frac{e^{ah}-1}{a}b\xi\Delta=0$ simultaneously.
System (\ref{con:multi_linear},\ref{con:linear1}) is uncontrollable if there exists some $\xi\in\mathbb{C}^{1\times{N}}$ and $s_0\in\mathbb{C}$, satisfying $s_0\xi=\xi(aI_N+cW)$ and $\xi\Delta=0$ simultaneously.
The latter condition is met by setting $s_0=a(1-\frac{s_1-e^{ah}}{1-e^{ah}})$, which completes the proof.
\end{proof}

\begin{corollary}\label{rem:1state}
If $0\notin\sigma(\Phi_s)$, the networked sampled-data system with one-dimensional node dynamics (\ref{MIMO_sampled_1},\ref{sampled1}) is controllable if and only if its original continuous-time system (\ref{con:multi_linear},\ref{con:linear1}) is controllable.
\end{corollary}

\begin{proof}
Since $\sigma(\Phi_s)=\{E_1,E_2,...,E_r\}$, $0\notin\sigma(\Phi_s)$ is equivalent to: $\forall{i}\in\{1,...,r\}$, $c\lambda_i\neq\frac{ae^{ah}}{(1-e^{ah})}$.
In this case, if $[(s_1-e^{ah})I_N-\frac{c}{a}(e^{ah}-1)W,\frac{b}{a}(e^{ah}-1)\Delta]$ is not of full row rank, system (\ref{MIMO_sampled_1},\ref{sampled1}) is uncontrollable.
Assume that the related continuous system (\ref{con:multi_linear},\ref{con:linear1}) is uncontrollable, then there exists some $\xi\in\mathbb{C}^{1\times{N}}$ and $s_0\in\mathbb{C}$ such that $s_0\xi=\xi(aI+cW)$ and $\xi\Delta=0$ simultaneously.
Let $s_1=1-\frac{s_0}{a}(1-e^{ah})$, it follows that $\xi[(s-e^{ah})I_N-\frac{c}{a}(e^{ah}-1)W,\frac{b}{a}(e^{ah}-1)\Delta]=[(s-e^{ah})\xi-s_0\frac{e^{ah}}{a}\xi,0]=0$,
which lead to the uncontrollability of system (\ref{MIMO_sampled_1},\ref{sampled1}).
\end{proof}

\subsection{Self-loop-state dynamics}

Consider a networked sampled-data system with identical multi-dimensional self-loop node dynamics, i.e., $A=I_n$.
In such a networked system, there are no interactions between the internal states of each node.
The dynamics will become 
\begin{equation}\label{con:loop}
  \Phi=I_{Nn}+W\otimes{HC}, \Psi=\Delta\otimes{B},
\end{equation}
for the continuous-time system (\ref{con:multi_linear}), 
and for the corresponding sampled-data system (\ref{MIMO_sampled_1}), 
\begin{equation}\label{sam:loop}
\begin{aligned}
  &\Phi_s=e^hI_{Nn}+(e^h-1)W\otimes{HC},\\
  &\Psi_s=(e^h-1)(\Delta\otimes{B}).
\end{aligned}
\end{equation}

\begin{corollary}\label{coro:loop}
The networked sampled-data system with self-loop-state node dynamics (\ref{MIMO_sampled_1},\ref{sam:loop}) is controllable if its original continuous-time system (\ref{con:multi_linear},\ref{con:loop}) is controllable.
\end{corollary}

\begin{proof}
Similar to the proof of Corollary \ref{thm:1state}, assume that the sampled-data system (\ref{MIMO_sampled_1},\ref{sam:loop}) is uncontrollable.
Then there exists some $\xi\in\mathbb{C}^{1\times{Nn}}$ and $s_1\in\mathbb{C}$ such that
$(e^h-1)\xi(W\otimes{HC})=(s_1-e^h)\xi$ and $(e^h-1)\xi(\Delta\otimes{B})=0$ simultaneously.
Let $s_0=\frac{s_1-e^h}{e^h-1}+1$, it follows that $\xi({W\otimes{HC}})=(s_0-1)\xi$
and
$\xi(\Delta\otimes{B})=0$ at the same time,
which indicates the corresponding continuous-time system (\ref{con:multi_linear},\ref{con:loop}) is also uncontrollable.
\end{proof}

\begin{corollary}
If $0\notin\sigma(\Phi_s)$, the networked sampled-data system with self-loop-state node dynamics (\ref{MIMO_sampled_1},\ref{sam:loop}) is controllable if and only if its original continuous-time system (\ref{con:multi_linear},\ref{con:loop}) is controllable.
\end{corollary}

\begin{proof}
Suppose that the continuous-time system (\ref{con:multi_linear},\ref{con:loop}) is uncontrollable, then there exist some $\xi\in\mathbb{C}^{1\times{Nn}}$ and $s_0\in\mathbb{C}$ satisfying
$\xi({W\otimes{HC}})=(s_0-1)\xi$ and
$\xi(\Delta\otimes{B})=0$ simultaneously.
Let $s_1=(e^h-1)s_0+1$, it follows that
$\xi(s_1-(e^hI_{Nn}+W\otimes\mathcal{H}(h)))=0$ and $\xi(\Delta\otimes\mathcal{B}(h))=(e^h-1)\xi\Delta\otimes{B}=0$ hold simultaneously.
Therefore, the corresponding sampled-data system (\ref{MIMO_sampled_1},\ref{sam:loop}) is also uncontrollable.\par
\end{proof}

\section{Networked multi-rate sampled-data systems}\label{sec:multi}
The sampling period on control channels may be different from that on transmission channels in the networked system.
According to the multiple relationships between the two sampling periods, the networked multi-rate sampled-data systems are divided into transmission multi-rate sampled-data (TMS) systems and control multi-rate sampled-data (CMS) systems.
The controllability of these two elementary types of multi-rate sampled-data systems is discussed in this section.

\subsection{Networked TMS systems}
Here, a networked TMS system means that all nodes are sampled periodically at the same on both transmission channels and control channels, but the control sampling period is an integer multiple of the transmission sampling period.
As illustrated in Figure \ref{fig:TMS}, in the TMS pattern $h$ is the transmission sampling period and $lh$ is the control sampling period, with $l$ being a positive integer.\par
The dynamics of node $i$ in the TMS system are represented as:
\begin{equation}
\left\{
        \begin{array}{l}
        \dot{x}_i(t)=Ax_i(t)+\sum_{j=1}^N{w_{ij}Hy_j((kl+r)h)}+\delta_iBu_i(klh)\\
        y_i((kl+r)h)=Cx_i((kl+r)h)\\
        \end{array}
\right.
\end{equation}
where $t\in((kl+r)h,(kl+r+1)h]$, $r=0,1,...,l-1$ and $k\in\mathbb{N}$.
The associated compact form can be written as
\begin{equation}\label{TMS_1}
    X((kl+r+1)h)=\Phi_s{X}((kl+r)h)+\Psi_s{U(klh)}.
\end{equation}

Furthermore, system (\ref{TMS_1}) can be transformed into a DLTI system with sampling period $lh$:
\begin{equation}\label{TMS}
    X((k+1)lh)=\tilde{\Phi}_s{X}(klh)+\tilde{\Psi}_s{U(klh)}\\
\end{equation}
with
\begin{equation}\label{TMS_detail}
    \tilde{\Phi}_s=\Phi_s^l, \tilde{\Psi}_s=(\Phi_s^{l-1}+\Phi_s^{l-2}+...+\Phi_s+I_{Nn})\Psi_s.
\end{equation}\par

By Lemmas \ref{lem1}-\ref{chara}, the eigenvalues and the corresponding eigenspace of $\Phi_s$ can be obtained.
Then the eigenvalues and eigenspace of $\tilde{\Phi}_s$ can also be calculated by Lemma \ref{eigenspace_TMS}.

\begin{lemma}\label{eigenspace_TMS}
Assume that $\sigma(\Phi_s)=\{\theta_1^1,...,\theta_1^{p_1},...\theta_r^1,...\theta_r^{p_r}\}$, then $\sigma(\tilde{\Phi}_s)=\{(\theta_1^1)^l,...,(\theta_1^{p_1})^l,...,(\theta_r^1)^l,...,(\theta_r^{p_r})^l\}$.
Moreover, if $\theta_i^j\neq{0}$, $M(\theta_i^j|\Phi_s)= M((\theta_i^j)^l|\tilde{\Phi}_s)$ for $i=1,2,...,r$, $j=1,2,...,p_i$.
Specially, if $(\theta_{i_1}^{j_1})^l=(\theta_{i_2}^{j_2})^l=...=(\theta_{i_q}^{j_q})^l\triangleq(\theta_i^j)^l$, where $i_k\in\{1,...,r\}$ and $j_k\in\{1,...,p_{i_k}\}$, $k=1,...,q$, $q>1$, $\odot_{k=1}^q{M(\theta_{i_k}^{j_k}|\Phi_s)}= M((\theta_i^j)^l|\tilde{\Phi}_s)$.
\end{lemma}

\begin{proof}
Given $\sigma(\Phi_s)=\{\theta_1^1,...,\theta_1^{p_1},...,\theta_r^1,...,\theta_r^{p_r}\}$ and $\tilde{\Phi}_s=\Phi_s^l$, $\sigma(\tilde{\Phi}_s)$ can be easily obtained according to the spectral mapping theorem.
In addition, $\forall\eta\in{M}(\theta_i^j|\Phi_s)$, $\eta\Phi_s^l=(\theta_i^j)^l\eta$, and therefore $\eta\in{M}((\theta_i^j)^l|(\Phi_s)^l)$.
Specially, assume that $(\theta_{i_1}^{j_1})^l=(\theta_{i_2}^{j_1})^l=...=(\theta_{i_q}^{j_q})^l\triangleq(\theta_i^j)^l$, and consider $\eta_1\in{M}(\theta_{i_1}^{j_1}|\Phi_s),...,\eta_q\in{M}(\theta_{i_q}^{j_q}|\Phi_s)$.
For any linear combination $\eta=\sum_{j=1}^q a_j\eta_j$, it follows that $\eta\Phi_s^l=\sum_{j=1}^q a_j\eta_j\Phi_s^l=(\theta_i^j)^l\sum_{j=1}^q a_j\eta_j=(\theta_i^j)^l\eta$, i.e., $\eta\in{M}((\theta_i^j)^l|\Phi_s^l)$.
Thus $\odot_{k=1}^q{M(\theta_{i_k}^{j_k}|\Phi_s)}\subset M((\theta_i^j)^l|\tilde{\Phi}_s)$.
\end{proof}

\begin{figure}[tb]
\centering
\begin{subfigure}[b]{.22\textwidth}
  \centering
  \includegraphics[width=\linewidth]{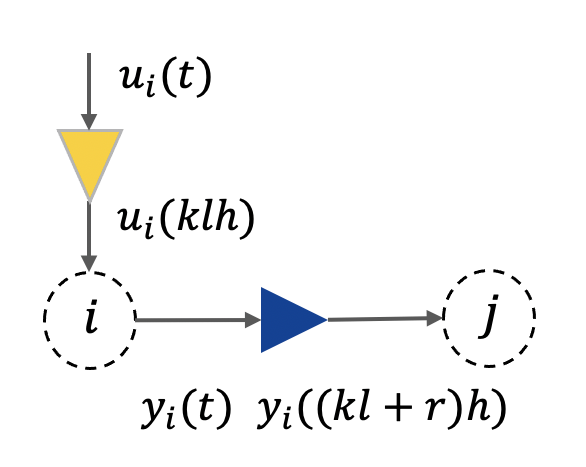}
  \caption{}
  \label{fig:TMS}
\end{subfigure}
\begin{subfigure}[b]{.22\textwidth}
  \centering
  \includegraphics[width=\linewidth]{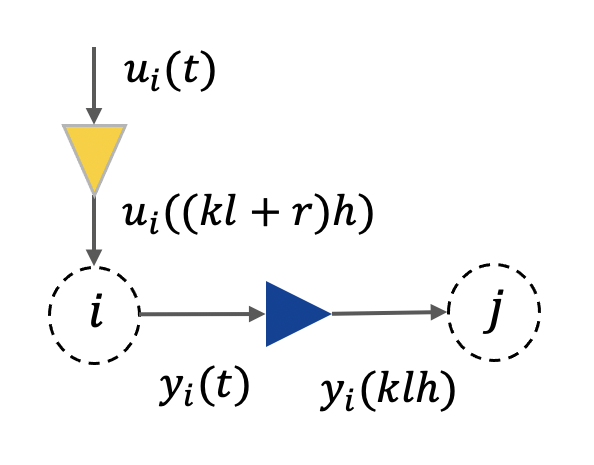} 
  \caption{}
  \label{fig:cms}
\end{subfigure}
\caption{Two elementary types of multi-rate sampling patterns. (a) TMS. (b) CMS.}
\label{fig:multi-rate}
\end{figure}

Based on Lemma \ref{eigenspace_TMS}, the controllability of the networked TMS system can be verified by Corollary \ref{TMS_con}.

\begin{corollary}\label{TMS_con}
The networked TMS system (\ref{TMS},\ref{TMS_detail}) is controllable if and only if the following conditions hold simultaneously:\par
(1) $\forall\theta_i^j\in\sigma(\Phi_s)$, $\theta_i^j\neq{0}$, $\sum_{c=0}^{l-1}{(\theta_i^j)^c}\neq{0},$ for every $i=1,...,r$ and $j=1,...,p_i$.\par
(2) $\forall\eta\in{M({\theta_i^j|\Phi_s}})$ and $\eta\neq{0}$, $\eta(\Delta\otimes\mathcal{B}(h))\neq{0}$, for every $i=1,...,r$ and $j=1,...,p_i$.
\end{corollary}

\begin{proof}
Corollary \ref{TMS_con} can be proved based on Theorem \ref{thm-main} and Lemma \ref{eigenspace_TMS}.
According to Lemma \ref{dlti}, system (\ref{TMS},\ref{TMS_detail}) is controllable if and only if $\forall\theta\in\sigma(\tilde{\Phi}_s)$, $\theta\neq{0}$, any eigenvector $\eta$ of $\tilde{\Phi}_s$ about $\theta$ satisfies $\eta\tilde{\Psi}_s\neq{0}$.
The above condition holds if and only if $\forall\theta_i^j\in\sigma(\Phi_s)$ and $\forall\eta\in{M(\theta_i^j|\Phi_s)}$, $\eta\tilde{\Psi}_s=\eta(\Phi_s^{l-1}+\Phi_s^{l-2}+...+\Phi_s+I_{Nn})\Psi_s=((\theta_i^j)^{l-1}+(\theta_i^j)^{l-2}+...+\theta_i^j+1)\eta(\Delta\otimes{\mathcal{B}(h)})\neq{0}$, i.e., condition (1) and (2) in Corollary \ref{TMS_con} hold simultaneously.\par
In addition, consider the case that $(\theta_{i_1}^{j_1})^l=(\theta_{i_2}^{j_2})^l=...=(\theta_{i_q}^{j_q})^l=\theta$, where $q>1$, $i_k\in\{1,2,...,r\}$, $j_k\in\{1,2,...,p_{i_k}\}$, $k=1,2,...,q$.
The controllability condition requires that $\forall\eta=a_1\eta_1+a_2\eta_2+...+a_q\eta_q$, with $[a_1,...,a_q]\neq{0}$, $a_k\in\mathbb{C}$, $\eta_k\in{M(\theta_{i_k}^{j_k}|\Phi_s)}$, $k=1,...,q$, satisfies
$\eta\tilde{\Psi}_s=\sum_{k=1}^qa_k\sum_{c=0}^{l-1}(\theta_{i_k}^{j_k})^c\eta_k(\Delta\otimes{\mathcal{B}(h)})\neq{0}.$
Since $\sum_{c=0}^{l-1}(\theta_{i_k}^{j_k})^c\neq{0}$ according to condition (1), $a_k\sum_{c=0}^{l-1}(\theta_{i_k}^{j_k})^c$ can be an arbitrary complex number, $k=1,...,q$.
It follows that $\sum_{k=1}^qa_k\sum_{c=0}^{l-1}(\theta_{i_k}^{j_k})^c\eta_k$ can be any vector in the eigenspace $\odot_{k=1}^q{M(\theta_{i_k}^{j_k}|\Phi_s)}$.
So far the proof of Corollary \ref{TMS_con} is complete.
\end{proof}

\subsection{Networked CMS systems}

Here, a networked CMS system means that all nodes are sampled periodically at the same time on both transmission channels and control channels, but the transmission sampling period is an integer multiple of the control sampling period.
As shown in Fig.\ref{fig:cms}, in the CMS pattern, $h$ is the control sampling period and $lh$ is the transmission sampling period, with $l$ being a positive integer.\par
The dynamics of the node in the networked CMS system are described as:
\begin{equation}
\left\{
        \begin{array}{l}
        \dot{x}_i(t)=Ax_i(t)+\sum_{j=1}^N{w_{ij}Hy_j(klh})+\delta_iBu_i((kl+r)h)\\
        y_i(klh)=Cx_i(klh)\\
        \end{array}
\right.
\end{equation}
where $t\in((kl+r)h,(kl+r+1)h]$, $r=0,1,...,l-1$ and $k\in\mathbb{N}$. 
The corresponding compact form can be written as
\begin{equation}\label{CMS_1}
    X((kl+r+1)h)=\Lambda{X}((kl+r)h)+\Gamma{X}(klh)+\Psi_s{U((kl+r)h)}\\
\end{equation}
with
\begin{equation}\label{CMS_2}
    \Lambda=I_N\otimes{e^{Ah}}, \Gamma=W\otimes\mathcal{H}(h).
\end{equation}\par

Furthermore, denote $\hat{U}(klh)=[U^\top(klh),U^\top((kl+1)h),$ ..., $U^\top((kl+l-1)h)]^\top$, then system (\ref{CMS_1},\ref{CMS_2}) can be transformed into a DLTI system with a sampling period $lh$:
\begin{equation}\label{CMS}
    X((k+1)lh)=\hat{\Phi}_s{X}(klh)+\hat{\Psi}_s{\hat{U}(klh)}\\
\end{equation}
with
\begin{equation}\label{CMS_detail}
\begin{aligned}
 \hat{\Phi}_s&=I_N\otimes{e^{Alh}}+W\otimes\mathcal{H}(lh),\\
    \hat{\Psi}_s&=[\Delta\otimes{e^{A(l-1)h}}\mathcal{B}(h),...,\Delta\otimes{e^{Ah}}\mathcal{B}(h),\Delta\otimes\mathcal{B}(h)].
\end{aligned}
\end{equation}\par

Based on Theorem \ref{thm2}, a sufficient condition for the controllability of the networked CMS system can be derived.
Denote $\hat{E}_i=e^{Alh}+\lambda_i\mathcal{H}(lh)$, $\sigma(\hat{E}_i)=\{\hat{\theta}_i^1,...,\hat{\theta}_i^{p_i}\}$, $i=1,2,...,r$.

\begin{corollary}\label{CMS_con}
The networked CMS system (\ref{CMS},\ref{CMS_detail}) is controllable if and only if $ \forall\eta\in{M(\hat{\theta}_i^j|\hat{\Phi}_s)}$, $\hat{\theta}_i^j\neq{0}$ and $\eta\neq{0}$, $\eta[\Delta\otimes{e^{A(l-1)h}}\mathcal{B}(h),...,\Delta\otimes{e^{Ah}}\mathcal{B}(h),\Delta\otimes\mathcal{B}(h)]\neq{0}$, for every $i=1,...,r$ and $j=1,...,p_i$.\par
\end{corollary}

\begin{proof}
The expression of $\hat{\Phi}_s$ is in the same form as $\Phi_s$, except that the sampling period of $\hat{\Phi}_s$ is $l$ times that of $\Phi_s$.
Therefore, Corollary \ref{CMS_con} can be proved based on Theorem \ref{lem1}.
\end{proof}

The above corollaries show that the controllability conditions of networked multi-rate sampled-data systems can be simplified based on the single-rate system (\ref{MIMO_sampled_1},\ref{MIMO_sampled_2}). 
As a result, verifying a networked multi-rate sampled-data system by Corollary \ref{TMS_con} or Corollary \ref{CMS_con} greatly reduces the computational complexity than using the PBH rank condition directly.
However, it has not been proved that the controllability of system (\ref{MIMO_sampled_1},\ref{MIMO_sampled_2}) can ensure that system (\ref{TMS},\ref{TMS_detail}) or system (\ref{CMS},\ref{CMS_detail}) is controllable, and vice versa.
But the above analysis has shed light on increasing the controllability of networked sampled-data systems by adjusting the multiple relationships between the sampling periods on control channels and transmission channels.

\section{Conclusion}\label{sec:con}
The controllability of networked sampled-data systems is investigated.
The sampling is periodic on the transmission and control channels, with single- and multi-rate patterns considered, respectively.
Necessary or/and sufficient controllability conditions are developed, indicating that the controllability of networked sampled-data systems is jointly determined by the external inputs, network topology, inner couplings, node dynamics, and the sampling period.
Results show that the pathological sampling of single node systems can be eliminated by an appropriate design of network topology and inner couplings.
However, when the topology matrix is singular, the pathological sampling of single node systems will inevitably lead to the loss of controllability of the whole system.
And for systems with one-dimensional or self-loop-state node dynamics, the sampling will not affect the controllability of the networked systems.
In further studies, we will consider more general systems with heterogeneous node dynamics.
In addition, more complex sampling patterns will be investigated.
More complex and deeper network structures will be studied, and the non-pathological sampling conditions of networked systems will be further explored.

\bibliography{reference}

\end{document}